\documentclass[11pt,letterpaper]{article}

\usepackage{xcolor}
\usepackage{newclude}
\usepackage{enumitem}
\usepackage{multirow}
\usepackage{hyperref}
\usepackage{geometry}
\usepackage{xspace}

\usepackage{amsthm,amsmath,amssymb} 		
\def\E{{\rm E}\,} 

\newtheorem{theorem}{Theorem}
\newtheorem{proposition}{Proposition}
\newtheorem{lemma}{Lemma}

\newcommand{\thmref}[1]{Theorem~\lowercase{\ref{#1}}}

\newcommand{\propref}[1]{Proposition~\lowercase{\ref{#1}}}

\def\clap#1{\hbox to 0pt{\hss#1\hss}}

\newcommand{\ie}{i.e.,\xspace}

\newcommand{\reals}{\mathbb{R}}
\newcommand{\Mbar}{\bar{M}}
\newcommand{\xbar}{\bar{x}}
\newcommand{\pbar}{\bar{p}}
\newcommand{\vtilde}{\tilde{v}}

\newcommand{\xtilde}{\tilde{x}}
\newcommand{\ptilde}{\tilde{p}}

\usepackage[numbers]{natbib}

\setenumerate{leftmargin=*}
\setenumerate[1]{label=(\alph*), ref=(\alph*)}

\usepackage{algpseudocode}
\usepackage{algorithm}
\algrenewcommand{\algorithmiccomment}[1]{\texttt{\textbackslash\textbackslash\ #1}}

\newcommand{\mc}[3]{\multicolumn{#1}{#2}{#3}}
\newcommand{\mr}[3]{\multirow{#1}{#2}{#3}}

\newcommand{\POall}{PO\@\xspace}

\newcommand{\cont}{\ (Continued)}

\newcommand{\ica}{$b_2 > v_{2,1} +v_{2,2}$}
\newcommand{\icb}{$v_{2,1}>v_{1,1}$}
\newcommand{\icbp}{$v_{1,1}>v_{2,1}$}
\newcommand{\icc}{$v_{2,2}>v_{1,2}$}
\newcommand{\iccp}{$v_{1,2}>v_{2,2}$}
\newcommand{\icd}{$b_1>v_{1,1}$}
\newcommand{\ice}{$v_{1,1}v_{2,2}>v_{1,2}v_{2,1}$}
\newcommand{\icf}{$v_{2,1} + v_{2,2} > b_1$}

\begin{document}


\title{Auctions with Heterogeneous Items and Budget Limits\thanks{%
  This work was funded by the Vienna Science and Technology Fund (WWTF)
  through project ICT10-002, by the University of Vienna through IK
  I049-N, and by a Google Research Award.}}

\author{%
  Paul D\"utting\thanks{%
  \'Ecole Polytechnique F\'ed\'erale de Lausanne, Switzerland,
  \texttt{paul.duetting@epfl.ch}}
  \and
  Monika Henzinger\thanks{%
  Fakult\"at f\"ur Informatik, Universit\"at Wien, Austria,
  \texttt{monika.henzinger@univie.ac.at}}
  \and
  Martin Starnberger\thanks{%
  Fakult\"at f\"ur Informatik, Universit\"at Wien, Austria,
  \texttt{martin.starnberger@univie.ac.at}}
}

\date{}

\maketitle


\begin{abstract}
We study individual rational, Pareto optimal, and incentive compatible 
mechanisms for auctions with heterogeneous items and budget limits. 
For multi-dimensional valuations we show that there can be no {\em 
deterministic} mechanism with these properties for {\em divisible} 
items. We use this to show that there can also be no {\em randomized} 
mechanism that achieves this for either {\em divisible} or {\em 
indivisible} items. For single-dimensional valuations we show that 
there can be no {\em deterministic} mechanism with these properties 
for indivisible items, but that there is a {\em randomized} mechanism 
that achieves this for either divisible or indivisible items.
The impossibility results hold for {\em public} budgets, while the 
mechanism allows {\em private} budgets, which is in both cases the 
harder variant to show. While all positive results are polynomial-time 
algorithms, all negative results hold independent of complexity 
considerations.
\end{abstract}


\section{Introduction}
A canonical problem in Mechanism Design is the design of economically efficient
auctions that satisfy individual rationality and incentive compatibility. In 
settings with quasi-linear utilities these goals are achieved by the 
Vickrey-Clarke-Groves (VCG) mechanism. In many practical situations, including
settings in which the agents have budget limits, the quasi-linear assumption 
fails to be true and, thus, the VCG mechanism is not applicable.

Ausubel \cite{Ausubel04} describes an ascending-bid auction for homogeneous 
items that yields the same outcome as the sealed-bid Vickrey auction, but 
offers advantages in terms of simplicity, transparency, and privacy preservation. 
In his concluding remarks he points out that ``when budgets impair the bidding 
of true valuations in a sealed-bid Vickrey auction, a dynamic auction may 
facilitate the expression of true valuations while staying within budget 
limits'' (p.~1469).

Dobzinski et al.~\cite{DobzinskiLaviNisan08} show that an adaptive version 
of Ausubel's ``clinching auction'' is indeed the unique mechanism that satisfies 
individual rationality, Pareto optimality, and incentive compatibility in 
settings with {\em public} budgets. They use this fact to show that there can 
be no mechanism that achieves those properties for {\em private} budgets.

An important restriction of Dobzinski et al.'s impossibility result for private
budgets is that it only applies to {\em deterministic} mechanisms. In fact, as  
Bhattacharya et al.~\cite{BhattacharyaEtAl10} show, there exists a {\em 
randomized} auction that is individual rational, Pareto optimal, and incentive
compatible with private budgets.

All these results assume that the items are homogeneous, although as Ausubel
\cite{Ausubel06} points out, ``situations abound in diverse industries in 
which heterogeneous (but related) commodities are auctioned'' (p.~602). He 
also describes an ascending-bid auction, the ``crediting and debiting auction'', 
that takes the place of the ``clinching auction'' when items are heterogeneous.

Positive and negative results for {\em deterministic} mechanisms and public 
budgets are given in \cite{FiatEtAl11,LaviMay11,GoelEtAl12,ColiniEtAl11}. We 
focus on {\em randomized} mechanisms, and prove positive results for private 
budgets and negative results for public budgets. {\em We thus explore the power 
and limitations of randomization in settings with heterogeneous items and budget 
limits.}

\paragraph{Model.}
There are $n$ agents and $m$ items. The items are either divisible or 
indivisible. Each agent has a valuation for each item and each agent has 
a budget. Agents can be assigned more than one item and valuations are 
additive across items. All valuations are private. We distinguish between 
settings in which budgets are public and settings in which budgets are 
private.
A mechanism is used to compute assignments and payments 
based on the reported valuations and the reported budgets. An agent's 
utility is defined as valuation for the assigned items minus the payment 
if the payment does not exceed the budget and the utility is minus infinity 
otherwise. We assume that agents are utility maximizers and as such need not 
report their true valuations and true budgets.

Our goal is to design mechanisms with certain desirable properties or to show 
that no such mechanism exists. For deterministic mechanisms we require that 
the respective properties are always satisfied. For randomized mechanisms we 
either require that the properties hold for all outcomes or that they hold in
expectation. In the former case they are satisfied {\em ex post}, in the latter 
they are satisfied {\em ex interim}. 

We are interested in the following properties:

(a) {\em Individual rationality (IR)}: A mechanism is IR if all outcomes it 
    produces give non-negative utility to the agents {\em and} the sum of the 
    payments is non-negative.
(b) {\em Pareto optimality (PO)}: A mechanism is PO if it produces an outcome 
    such that there is no other outcome in which all agents and the auctioneer 
    are no worse off and at least one of the agents or the auctioneer is 
    strictly better off.
    \footnote{If the outcome for which we want to establish PO is IR, then we 
    only have to consider alternative outcomes that are IR. In the alternative
    outcome individual payments may be negative, even if the original outcome 
    satisfied IR and NPT. See the arXiv version of \cite{FiatEtAl11} for a more 
    detailed discussion.}
(c) {\em No positive transfers (NPT)}: A mechanism satisfies NPT if it produces 
    an outcome in which all payments are non-negative.
(d) {\em Incentive compatibility (IC)}: A mechanism is IC if each agent 
    maximizes his utility by reporting his true valuation(s) and true budget no 
    matter what the other agents' reported valuations and reported budgets are.
    If the budget is public then the agents can only report their true budgets.

Following prior work we focus on IR, PO, NPT, and IC for positive results and 
on IR, PO, and IC for negative results. Both the inclusion of NPT for positive 
results and the exclusion of NPT for negative results strengthens the respective
results. 

\paragraph{Results.}
We analyze two settings with heterogeneous items, one with multi-dimensional 
valuations and one with single-dimensional valuations. In the setting with 
multi-dimensional valuations, each agent has an arbitrary, non-negative 
valuation for each of the items. In the setting with single-dimensional 
valuations, which is inspired by sponsored search auctions, an agent's 
valuation for an item is the product of an item-specific quality and an 
agent-specific valuation. Our motivation for studying this setting is that 
an advertiser might want to show his ad in multiple slots on a search result 
page.

(a) For {\bf  multi-dimensional} valuations the impossibility result of 
\cite{FiatEtAl11} implies that there can be no deterministic mechanism 
for {\em  indivisible} items that is IR, PO, and IC for public budgets. 
We show that there also can be no deterministic mechanism with these 
properties for {\em divisible} items. We use this to show that for both 
divisible and indivisible items there can be no randomized mechanism
that is IR ex interim, PO ex interim, and IC ex interim. 
This is the first impossibility result for randomized mechanisms for 
auctions with budget limits. It establishes an interesting separation 
between randomized mechanisms for single-dimensional valuations, where 
such mechanisms exist (see below), and multi-dimensional valuations,
where no such mechanism exists.

(b) For {\bf single-dimensional} valuations the impossibility result of 
\cite{DobzinskiLaviNisan08} implies that there can be no deterministic 
mechanism for indivisible items that is IR, PO, and IC for {\em private} 
budgets. We show that for heterogeneous items there can also be no 
deterministic mechanism for indivisible items that is IR, PO, and IC for 
{\em public} budgets. We thus obtain a strong separation between
deterministic mechanisms, that do {\em not} exist for {\em public} budgets, 
and randomized mechanisms, that exist for {\em private} budgets (see below). 
This separation is stronger than in the homogeneous items setting, where 
a deterministic mechanism exists for public budgets \cite{DobzinskiLaviNisan08}.
Additionally, our impossibility result is tight in the sense that if any
of the conditions is relaxed such a mechanism exists: (i) For 
{\em homogeneous}, indivisible items a deterministic mechanism is given 
in \cite{DobzinskiLaviNisan08}, (ii) we give a deterministic mechanism for 
heterogeneous, {\em divisible} items, and (iii) we give a {\em randomized} 
mechanism for heterogeneous, indivisible items.

(c) For {\bf single-dimensional} valuations we give mechanisms that extend 
earlier work for homogeneous items to heterogeneous items. Specifically, we 
give a {\em randomized} mechanism that satisfies IR ex interim, NPT ex post, 
PO ex post, and IC ex interim for {\em divisible} or {\em indivisible} items 
and  {\em public} or {\em private} budgets. 
Additionally, for the case of {\em divisible} items and {\em public} budgets 
we give a  {\em deterministic} mechanism that is IR, NPT, PO, and IC.

We summarize our results and the results from related work described next 
in Table~\ref{tab:summary_indivisible} and Table~\ref{tab:summary_divisible} 
below.

\paragraph{Related Work.}
The setting in which all items are identical was first studied by 
\cite{DobzinskiLaviNisan08}. By adapting the ``clinching auction'' 
of \cite{Ausubel04} from settings without budgets to settings with
budgets they obtain deterministic mechanisms that are IR, NPT, PO, 
and IC with public budgets for divisible and indivisible items. They also 
show that these mechanisms are the only mechanisms that are IR, PO, and IC, 
and that they are not IC for private budgets, implying that there can be no 
deterministic mechanism that is IR, PO, and IC when the budgets are private. 
However, \cite{BhattacharyaEtAl10} showed that there is such a mechanism 
for private budgets that is randomized. Note that both, 
\cite{DobzinskiLaviNisan08} and \cite{BhattacharyaEtAl10} study only 
homogeneous items. 

Impossibility results for general, non-additive valuations were given in 
\cite{LaviMay11,ColiniEtAl11,GoelEtAl12}. Combined they show that 
there can be no deterministic mechanism for indivisible items that is IR, PO, 
and IC with public budgets for monotone valuations with decreasing marginals. 
These impossibility results do not apply to additive valuations, which is 
the case that we study. 

Heterogeneous items were first studied in \cite{FiatEtAl11}. 
In their model each agent has the same valuation for 
each item in an agent-dependent interest set and zero for all other items. 
They give a deterministic mechanism for indivisible items that satisfies IR, NPT, 
PO, and IC when both interest sets and budgets are public. They also show that 
when the interest sets are \emph{private}, then there can be no
deterministic mechanism that satisfies IR, PO, and IC. This implies that for 
\emph{indivisible} items and public budgets there can be no 
\emph{deterministic} IR, PO, and IC mechanism for unconstrained valuations.

Settings with heterogeneous items were in parallel to this paper studied 
by \cite{ColiniEtAl11} and \cite{GoelEtAl12}. The former 
study problems with multiple keywords, 
each having multiple slots. Agents have unit demand per keyword. They are 
either interested in a subset of the keywords and have identical valuations 
for the slots or they are interested in all keywords and have sponsored 
search like valuations for the slots. The latter study settings in which the 
agents have identical valuations and the allocations must satisfy polymatroidal 
or polyhedral constraints.

The settings studied in \cite{ColiniEtAl11,GoelEtAl12} are more general than 
the single-dimensional valuations setting studied here. 
On the one hand this implies that their positive results apply to the 
single-dimensional valuations setting studied here, and show that there are 
deterministic mechanisms for divisible items and randomized mechanisms for 
both divisible and indivisible items that are IC with \emph{public} budgets.
On the other hand this implies that our negative result for the 
single-dimensional valuations setting applies to the settings studied in 
these papers, and shows that there can be no 
{\it deterministic} mechanisms that are IC with {\it public} budgets for 
{\it indivisible} items.
Finally, the impossibility results presented in \cite{ColiniEtAl11,GoelEtAl12} 
either assume that the valuations are non-additive or that the allocations 
satisfy arbitrary polyhedral constraints and have therefore no implications 
for the multi-dimensional valuations setting studied here.

\paragraph{Overview.}
We summarize the results from related work and this paper for indivisible items 
in Table~\ref{tab:summary_indivisible} and  for divisible items in 
Table~\ref{tab:summary_divisible}. We use a plus ($+$ or $\oplus$) to indicate 
that there is an IR, PO, NPT, and IC mechanism. We use a minus ($-$ or $\ominus$) 
to indicate that there is no IR, PO, and IC mechanism. We use $+$ and $-$ for 
results from related work and $\oplus$ and $\ominus$ for results from this paper. 
A question mark ($?$) indicates that nothing is known for this setting. For the 
model of \cite{FiatEtAl11} the table has two entries, one for public and one 
for private interest sets. {\em While all positive results from this paper are 
polynomial-time algorithms, all negative results hold independent of complexity 
considerations.}


\begin{table}[ht]
{
\small
\caption{Results for {\em Indivisible} Items from Related Work and this Paper}
{
\begin{tabular}{cccccccc}
\hline
&       &\mc{2}{c}{homogeneous}&\mc{4}{c}{heterogeneous \& additive}           \\
&&\mr{2}{*}{add.}&\mr{2}{*}{non-add. }&interest set&multi-keyword&\mr{2}{*}{single-dim.}&\mr{2}{*}{multi-dim.}\\
&budgets&&&public/private&unit demand&&\\
\hline
\mr{2}{*}{det.} &public  &$+$ \cite{DobzinskiLaviNisan08}&$-$\cite{LaviMay11,ColiniEtAl11}&$+$\cite{FiatEtAl11}/$-$\cite{FiatEtAl11}&$\ominus$&$\ominus$&$-$ \cite{FiatEtAl11}\\
                &private &$-$ \cite{DobzinskiLaviNisan08}&$-$ \cite{DobzinskiLaviNisan08}&$-$\cite{DobzinskiLaviNisan08}/$-$\cite{DobzinskiLaviNisan08}&$-$ \cite{DobzinskiLaviNisan08}&$-$\cite{DobzinskiLaviNisan08}&$-$ \cite{DobzinskiLaviNisan08}\\
\mr{2}{*}{rand.}&public  &$+$ \cite{DobzinskiLaviNisan08}&?&$+$\cite{FiatEtAl11}/?&$+$\cite{ColiniEtAl11,GoelEtAl12}&$\oplus$&$\ominus$\\
                &private &$+$ \cite{BhattacharyaEtAl10}&?&?/?&?&$\oplus$&$\ominus$\\
\hline
\end{tabular}
}
}
\label{tab:summary_indivisible}
\end{table}

\vspace{-12pt}

\begin{table}[ht]
{
\small
\caption{Results for {\em Divisible} Items from Related Work and this Paper}
{
\begin{tabular}{cccccccc}
\hline
&\mc{1}{c}{      }&\mc{2}{c}{homogeneous}&\mc{4}{c}{heterogeneous \& additive}\\
&&\mr{2}{*}{add.}&\mr{2}{*}{non-add. }&polymatroid&multi-keyword&\mr{2}{*}{single-dim.}&\mr{2}{*}{multi-dim.}\\
&budgets&&&constraints&unit demand&&\\
\hline
\mr{2}{*}{det.} &public  &$+$ \cite{DobzinskiLaviNisan08,BhattacharyaEtAl10}&$-$\cite{GoelEtAl12}&$+$\cite{GoelEtAl12}&$+$\cite{ColiniEtAl11,GoelEtAl12}&$\oplus$&$\ominus$\\
                &private &$-$ \cite{DobzinskiLaviNisan08}&$-$\cite{DobzinskiLaviNisan08}&$-$ \cite{DobzinskiLaviNisan08}&$-$ \cite{DobzinskiLaviNisan08}&$-$ \cite{DobzinskiLaviNisan08}&$-$ \cite{DobzinskiLaviNisan08}\\
\mr{2}{*}{rand.}&public  &$+$ \cite{DobzinskiLaviNisan08,BhattacharyaEtAl10}&?&$+$\cite{GoelEtAl12}&$+$\cite{ColiniEtAl11,GoelEtAl12}&$\oplus$&$\ominus$\\
                &private &$+$ \cite{BhattacharyaEtAl10}&?&?&?&$\oplus$&$\ominus$\\
\hline
\end{tabular}
}
}
\label{tab:summary_divisible}
\end{table}

\vspace{-8pt}

\paragraph{Techniques.}
Our technical contributions are as follows:

(a) For multi-dimensional valuations we obtain a partial characterization of IC
    by generalizing the ``weak  monotonicity'' (WMON) condition of 
    \cite{BikhchandaniEtAl06} from settings {\em without budgets} to settings 
    {\em with public budgets}. We obtain our impossibility result for 
    deterministic mechanisms and divisible items by showing that in certain 
    settings WMON will be violated.
    For this we use that multi-dimensional valuations enable the agents to lie 
    in a sophisticated way: While all previous impossibility 
    proofs in this area used agents that either only overstate or only 
    understate their valuations, we use an agent that overstates his valuation 
    for one item and understates his valuation for another. 

(b) For single-dimensional valuations and both divisible and indivisible 
    items we 
    characterize PO by a simpler ``no trade'' (NT) condition. Although this
    condition is more complex than similar conditions in 
    \cite{DobzinskiLaviNisan08,BhattacharyaEtAl10,FiatEtAl11}, we are able to 
    show that an outcome is PO if and only if it satisfies NT.
    We also generalize the ``classic''characterization results of IC mechanism
    of \cite{Myerson81,ArcherTardos01} from settings {\em without budgets} to 
    settings with {\em public budgets} by showing that a mechanism is IC with 
    public budgets if and only if it satisfies ``value monotonicity'' (VM) and 
    ``payment identity'' (PI). The characterizations of PO and IC with public
    budgets play a crucial role in the proof of our impossibility result for   
    indivisible items, which uses NT and PI to derive lower bounds on the 
    agents' payments that conflict with the upper bounds on the payments 
    induced by IR. 

(c) We establish the positive results for single-dimensional valuations and both
    divisible and indivisible items by giving a new reduction of this case to 
    the case of a single and by definition homogeneous item. This allows us to 
    apply the techniques that \cite{BhattacharyaEtAl10} developed for the 
    single-item setting. This is a general reduction between the heterogeneous 
    items setting and the homogeneous items setting, which is likely to have 
    further applications.

(d) We give an explicit polynomial-time algorithm for the ``adaptive clinching
    auction'' for {\em divisible} items and an arbitrary number of agents.
    To the best of our knowledge we are the first ones to actually give a 
    polynomial-time version of this auction for arbitrarily many agents.


\section{Problem Statement}
We are given a set $N$ of $n$ agents and a set $M$ of $m$ items. We 
distinguish between settings with divisible items and settings with 
indivisible items. In both settings we use $X = \prod_{i=1}^{n} X_i$ for the 
allocation space. For divisible items the allocation space is $X_i 
= [0,1]^m$ for all agents $i \in N$ and $x_{i,j} \in [0,1]$ denotes the 
fraction of item $j \in M$ that is allocated to agent $i\in N$. 
For indivisible items the allocation space is $X_i = \{0,1\}^m$ for all agents 
$i \in N$ and $x_{i,j} \in \{0,1\}$ indicates whether item $j \in M$ is 
allocated to agent $i \in N$ or not. In both cases we require that 
$\sum_{i=1}^{n} x_{i,j} \le 1$ for all items $j \in M$. We do {\em not} require 
that $\sum_{j=1}^{m} x_{i,j} \le 1$ for all agents $i \in N$, \ie we do {\em 
not} assume that the agents have unit demand.

Each agent $i$ has a type $\theta_i = (v_i,b_i)$ consisting of a valuation 
function $v_i: X_i \rightarrow \reals_{\ge 0}$ and a budget $b_i \in 
\reals_{\ge 0}$. We use $\Theta = \prod_{i=1}^{n} \Theta_i$ for the type 
space. We consider two settings with heterogeneous items, one with multi-
and one with single-dimensional valuations. In the first setting, each agent 
$i \in N$ has a valuation $v_{i,j} \in \reals_{\ge 0}$ for each item $j \in M$ 
and agent $i$'s valuation for allocation $x_i$ is $v_i(x_i) = \sum_{j=1}^{m} 
x_{i,j} v_{i,j}$. In the second setting, which is inspired by sponsored search 
auctions, each agent $i \in N$ has a valuation $v_i \in \reals_{\ge 0}$, each 
item $j \in M$ has a quality $\alpha_j \in \reals_{\ge 0}$, and agent $i$'s 
valuation for allocation $x_i \in X_i$ is $v_{i}(x_i) = \sum_{j=1}^{m} x_{i,j} 
\alpha_j v_i$. For simplicity we will assume that in this setting $\alpha_1 > 
\alpha_2 > \dots > \alpha_m$ and that $v_1 > v_2 > \dots > v_n > 0$.

A (direct revelation) mechanisms $M = (x,p)$ consisting of an allocation rule 
$x: \Theta \rightarrow X$ and a payment rule $p: \Theta \rightarrow \reals^n$ 
is deployed to compute an outcome $(x,p)$ consisting of an allocation $x\in X$ 
and payments $p \in \reals^n.$ 
We say that a mechanism is deterministic if the computation of $(x,p)$ is 
deterministic, and it is randomized if the computation of $(x,p)$ is randomized.

We assume that the agents are utility maximizers and as such need not report 
their types truthfully. We consider settings in which both the valuations and 
budgets are private and settings in which only the valuations are private and 
the budgets are public. When the valuations resp.~budgets are private, then the 
other agents have no knowledge about them, not even about their distribution. 
In the former setting a report by agent $i \in N$ with true type $\theta_i = (v_i,
b_i)$ can be any type $\theta'_i=(v'_i,b'_i).$ In the latter setting agent $i 
\in N$ is restricted to reports of the form $\theta'_i = (v'_i,b_i).$ In both 
settings, if mechanism $M=(x,p)$ is used to compute an outcome for reported 
types $\theta' = (\theta'_1, \dots, \theta'_n)$ and the true types are $\theta = 
(\theta_1, \dots, \theta_n)$ then the utility of agent $i \in N$ is 
\begin{align*}
  u_i(x_i(\theta'), p_i (\theta' ),\theta_i) = 
   \begin{cases}
     v_{i}(x_i(\theta'))-p_i(\theta') & \text{if $p_i(\theta') \le b_i$, and}\\
     - \infty                         & \text{otherwise.}
   \end{cases}
\end{align*}

For deterministic mechanisms and their outcomes we are interested in the 
following properties:

(a) {\em Individual rationality (IR)}: A mechanism is IR if it always produces
    an IR outcome. An outcome $(x,p)$ for types $\theta = (v,b)$ is IR if it 
    is (i) {\em agent rational:} $u_i(x_i,p_i,\theta_i) \ge 0$ for all agents 
    $i \in N$ and (ii) {\em auctioneer rational:} $\sum_{i=1}^{n} p_i \ge 0$.
(b) {\em Pareto optimality (PO)}: A mechanism is PO if it always produces a PO
    outcome. An outcome $(x,p)$ for types $\theta = (v,b)$ is PO if there is no 
    other outcome $(x',p')$ such that $u_i(x'_i,p'_i,\theta_i) \ge u_i(x_i,p_i,
    \theta_i)$ for all agents $i \in N$ and $\sum_{i=1}^{n} p'_i \ge 
    \sum_{i=1}^{n} p_i$, with at least one of the inequalities strict.\footnote{%
    Both IR and PO are defined with respect to the reported types, and are 
    satisfied with respect to the true types only if the mechanism also satisfies IC.}
(c) {\em No positive transfers (NPT)}: A mechanism satisfies NPT if it always 
    produces an NPT outcome. An outcome $(x,p)$ satisfies NPT if $p_i \ge 0$ 
    for all agents $i \in N.$
(d) {\em Incentive compatibility (IC)}: A mechanism satisfies IC if for all 
    agents $i \in N$, all true types $\theta$, and all reported types $\theta'$ 
    we have $u_i(x_i(\theta_i,\theta'_{-i}), p_i(\theta_i,\theta'_{-i}),
    \theta_i) \ge u_i(x_i(\theta'_i, \theta'_{-i}), p_i(\theta'_i,\theta'_{-i}),
    \theta_i).$

If a randomized mechanism satisfies any of these conditions in expectation, 
then we say that the respective property is satisfied {\em ex interim}. If
it satisfies any of these properties for all outcomes it produces, then we
say that it satisfies the respective property {\em ex post}.


\section{Multi-Dimensional Valuations}\label{sec:multidim}
In this section we obtain a partial characterization of mechanisms that are IC 
with public budgets by generalizing the ``weak monotonicity'' condition of 
\cite{BikhchandaniEtAl06} from settings without budgets to settings with budgets. 
We use this partial characterization together with a sophisticated way of 
lying, in which an agent understates his valuation for some item and overstates
his valuation for another item, to prove that there can be no {\em deterministic} 
mechanism for {\em divisible} items that is IR, PO, and IC with public budgets. 
Afterwards, we use this result to show that there can be no {\em randomized} 
mechanism for either {\em divisible} or {\em indivisible} items that is IR ex 
interim, PO ex interim, and IC ex interim for public budgets.

\paragraph{Partial Characterization of IC.}
For settings {\em without budgets} every mechanism that is incentive compatible
must satisfy what is known as {\em weak monotonicity (WMON)}, namely if $x'_i$ and 
$x_i$ are the assignments of agent $i$ for reports $v'_i$ and $v_i$, then 
the difference in the valuations for the two assignments must be at least as 
large under $v'_i$ as under $v_i$, i.e., $v'_i(x_i(\theta'_i, \theta_{-i})) - 
v'_i(x_i(\theta_i,\theta_{-i})) \ge v_i(x_i(\theta'_i, \theta_{-i})) - 
v_i(x_i(\theta_i,\theta_{-i})).$ We show that this is also true for 
mechanisms that respect the publicly known budget limits.\footnote{Without 
this restriction we could charge $p_i > b_i$ from all agents $i \in N$ to be IC. 
This restriction is satisfied by IR mechanisms to which we will apply this result.}

\begin{proposition} \label{prop:char-multi-dim}
If a mechanism $M=(x,p)$ for multi-dimensional valuations and either divisible 
or indivisible items that respects the publicly known budget limits is IC, then 
it satisfies WMON.
\end{proposition}

\paragraph{Deterministic Mechanisms for Divisible Items.}
We prove the impossibility result by analyzing a setting with two agents and two 
items. This restriction is without loss of generality 
as the impossibility result for an arbitrary number of agents $n > 2$ and an 
arbitrary number of items $m > 2$ follows by setting $v_{i,j} = 0$ if $i > 2$ 
or $j > 2$. In our impossibility proof agent~2 is not budget restricted (i.e., 
$b_2 > v_{2,1} + v_{2,2}$). Agents can lie when they report their valuations, 
and it is not sufficient to study a single input to prove the impossibility.  
Hence, we study the outcome for three related cases, namely Case~1 where  
$v_{1,1} < v_{2,1}$ and $v_{1,2} < v_{2,2}$; Case~2 where  $v_{1,1} > v_{2,1}$, 
$v_{1,2} < v_{2,2}$, and $b_1>v_{1,1}$; and Case~3 where  $v_{1,1} > v_{2,1}$, 
$v_{1,2} > v_{2,2}$, and additionally, \icd, \ice, and \icf. We give a partial 
characterization of those cases, which allows us to analyze the rational 
behavior of the agents. 

Case~1 is easy: Agent 2 is not budget restricted and has the highest valuations 
for both items; so he will get both items. Thus the utility for agent 1 is zero. 
Based on this observation Case~2 can be analyzed: Agent 1 has the higher 
valuation for item 1, while agent 2 has the higher valuation for item 2. Thus, 
agent 1 gets item 1 and agent 2 gets item 2. Since the only difference to Case~1 
is that in Case~2 $v_{1,1} > v_{2,1}$ while in Case~1 $v_{1,1} < v_{2,1}$, the 
critical value whether agent 2 gets item 1 or not is $v_{2,1}$. Thus, in every 
IC mechanism, agent 1 has to pay $v_{2,1}$ and has utility $v_{1,1} - v_{2,1}$. 
The details of these proofs can be found in Appendix~\ref{app:C1C2}.
Using these observations we are able to exactly characterize the 
allocation produced in Case~3 as follows: In Case~3 agent~1 has a higher 
valuation than agent~2 for both items, but
he does not have enough budget to pay for both fully. 
First we show that if agent~1 does not spend his whole budget ($p_1 < 
b_1$) he must fully receive both items (specifically $x_{1,2}=1$), since if 
not, he would buy more of them. Additionally, even if he spent his budget fully 
(i.e., $p_1 = b_1$) his utility $u_i$, which equals $x_{1,1} v_{1,1} + 
x_{1,2} v_{1,2}  - b_1$, must be non-negative. Since $b_1 > v_{1,1}$ this 
implies that $x_{1,1}$ must be 1, i.e., he must receive item 1 fully, and 
$x_{1,2}$ must be non-zero.

\begin{lemma}\label{lem:mult-case3a}
Given \icbp, \iccp, \icd, and \ice, if $p_{1}<b_{1}$ 
then $x_{1,1}=1$ and $x_{1,2}=1$, else if $p_{1}=b_1$ then $x_{1,1}=1$ and 
$x_{1,2}>0$, in every IR and \POall outcome.
\end{lemma}

Then we show that actually $x_{1,2} < 1$, which, combined with the previous 
lemma, implies that $p_1 = b_1$. The fact that $x_{1,2} < 1$, i.e, that 
agent 1 does not fully get item 1 {\em and} 2 is not surprising since he 
does not have enough budget to outbid agent 2 on both items as $b_1 < 
v_{2,1} + v_{2,2}$. However, we are even able to determine the exact value of 
$x_{1,2}$, which is $(b_1 - v_{2,1}) / v_{2,2}$.

\begin{lemma}\label{lem:mult-case3b}
Given \ica, \icbp, \iccp, \icd, \ice, and \icf, then 
$p_1=b_1$ and $x_{1,2}=(b_1 - v_{2,1})/v_{2,2} <1$ in every IR and \POall 
outcome selected by an IC mechanism.
\end{lemma}

We combine these characterizations of Case~3 with (a) the WMON property shown in
\propref{prop:char-multi-dim} and (b) a sophisticated way of the agent to lie: 
He {\em overstates} his value for
item 1 by a value $\alpha$ and {\em understates} his value for item 2 by a 
value $0< \beta < \alpha$, but by
such small values that Case~3 continues to hold.
Thus, by Lemma~\ref{lem:mult-case3a} $x_{2,1}$ remains 0 (whether the agent 
lies or does not), and thus, the
WMON condition implies that $x_{2,2}$ does {\em not} increase. 
However, by the dependence of $x_{1,2}$ on $v_{2,1}$ and $v_{2,2}$ shown in 
Lemma~\ref{lem:mult-case3b}, $x_{1,2}$, and thus also $x_{2,2}$ changes when 
agent 2 lies. This gives a contradiction to the assumption that such a 
mechanism exists.

\begin{theorem}\label{thm:impossibility-mult-dim}
There is no deterministic IC mechanism for divisible items which selects for any 
given input with public budgets an IR and \POall outcome.
\label{thm:multidim-impossibility}
\end{theorem}
\begin{proof}
Let us assume by contradiction that such a mechanism exists and consider an input 
for which \ica, \icbp, \iccp, \icd, \ice, and \icf\ holds. 
Such an input exists,
for example $v_{1,1} = 4$, $v_{1,2} = 5$, $v_{2,1} = 3$, and $v_{2,2} = 4$ with budgets
$b_1 = 5$ and $b_2 = 8$ would be such an input.
Lemma~\ref{lem:mult-case3a} and~\ref{lem:mult-case3b} 
imply that $x_{1,1}=1$, $x_{2,1}=0$, $x_{1,2}=\frac{b_1-v_{2,1}}{v_{2,2}}$, 
$x_{2,2}=1-x_{1,2}$, and $p_1 = b_1$.
Let us consider an alternative valuation by agent~2. We define $v'_{2,1}=
v_{2,1}+\alpha$ and $v'_{2,2}=v_{2,2}-\beta$ for arbitrary $\alpha, \beta>0$ and 
$\alpha > \beta$ which are sufficiently small such that $v_{1,1} v_{2,2}' > 
v_{1,2}v_{2,1}'$ holds. By 
\propref{prop:char-multi-dim}, IC implies WMON, and therefore, $x'_{2,2} 
v'_{2,2} - x_{2,2} v'_{2,2} \geq x'_{2,2} v_{2,2} - x_{2,2} v_{2,2}$. It follows 
that $x_{2,2}\geq x'_{2,2}$, and by Lemma~\ref{lem:mult-case3b}, 
$\frac{b_1-v_{2,1}}{v_{2,2}}   \leq \frac{b_1 -v'_{2,1}}{v'_{2,2}}$. Hence, the 
budget of agent~1 has to be large enough, such that $b_1\geq\frac{v_{2,2} 
v'_{2,1}-v_{2,1} v'_{2,2}}{v_{2,2}-v'_{2,2}}=\frac{v_{2,1}\beta+
v_{2,2}\alpha}{\beta}>v_{2,1}+v_{2,2}$, but $b_1<v_{2,1}+v_{2,2}$ holds by 
assumption. Contradiction!
\end{proof}

\paragraph{Randomized Mechanisms for Divisible and Indivisible Items.}
We exploit the fact that randomized mechanisms for both divisible and indivisible 
items are essentially equivalent to deterministic mechanisms for divisible items. 

We show that for agents with budget constraints every randomized mechanism 
$\Mbar=(\bar{x},\bar{p})$ for divisible or indivisible items can be mapped 
bidirectionally to a deterministic mechanism $M=(x,p)$ for divisible items with 
identical expected utility for all the agents and the auctioneer when the same 
reported types are used as input. To turn a randomized mechanism for 
{\em indivisible} items into a deterministic mechanism for {\em divisible} 
items simply compute the expected values of $p_i$ and 
$x_{i,j}$ for all $i$ and $j$ and return them. To turn a deterministic mechanism
for {\em divisible} items into a randomized mechanism for {\em indivisible} 
items simply pick values with probability $x_{i,j}$ and keep the same payment 
as the deterministic mechanism.

\begin{proposition}\label{prop:divisible=randomized}
Every randomized mechanism $\Mbar=(\bar{x},\bar{p})$ for agents with finite 
budgets, a rational auctioneer, and a limited amount of divisible 
or indivisible items can be mapped bidirectionally to a deterministic mechanism 
$M=(x,p)$ for divisible items such that $u_i(x_i(\theta'), p_i(\theta'),\theta_i) 
= \E[u_i(\xbar_i(\theta'),\pbar_i(\theta'),\theta_i)]$ and $\sum_{i\in N} 
p_i(\theta')=\E[\sum_{i\in N} \pbar_i(\theta')]$ for all agents $i$, all true
types $\theta=(v,b)$, and reported types $\theta'=(v',b')$. 
\end{proposition}
\begin{proof}
Let us map $\Mbar=(\bar{x},\bar{p})$ to 
$M=(x,p)$ that assigns for each agent $i \in N$ and item $j \in M$ a fraction of 
$\E[\xbar_{i,j}]$ of item $j$ to agent $i$, and makes each agent $i \in N$ pay 
$\E[\pbar_i]$. The expectations exist since the feasible fractions of items and 
the feasible payments have an upper bound and a lower bound. For the other 
direction, we map $M=(x,p)$ to 
$\Mbar=(\xbar,\pbar)$ that randomly picks for each item $j \in M$ an 
agent $i \in N$ to which it assigns item $j$ in a way such that agent $i$ is 
picked with probability $x_{i,j}$, and makes each agent $i \in N$ pay $p_i$. 
Since $x = \E[\xbar]$ and $p=\E[\pbar]$, $\sum_{j\in M} (x_{i,j}v_{i,j}) - p_i = 
\E[\sum_{j\in M} (\xbar_{i,j} v_{i,j}) - \pbar_i]$ for all $i\in N$ and 
$\sum_{i\in N} p_i=\E[\sum_{i\in N} \pbar_i]$.
\end{proof}

This proposition implies the non-existence of randomized mechanisms stated in 
Theorem~\ref{thm:randpoweak}.

\begin{theorem}\label{thm:randpoweak}
There can be no randomized mechanism for divisible or indivisible 
items that is IR ex interim, PO ex interim, and
IC ex interim, and that satisfies the public budget constraint ex post.
\end{theorem}
\begin{proof}
For a contradiction suppose that there is such a randomized mechanism. Then, by 
\propref{prop:divisible=randomized}, 
there must be a deterministic mechanism for divisible items and public budgets 
that satisfies IR, \POall, and IC. This gives a contradiction to 
\thmref{thm:multidim-impossibility}.
\end{proof}


\section{Single-Dimensional Valuations}
In this section we present exact characterizations of PO outcomes and mechanisms 
that are IC with public budgets. We characterize PO by a simpler ``no trade'' 
condition and, similar to Section~\ref{sec:multidim}, we extend the ``classic''
characterization results for IC mechanisms for single-dimensional valuations 
(see, e.g., \cite{Myerson81,ArcherTardos01}) without budgets to settings 
with public budgets. We use these characterizations to show that there can 
be no deterministic mechanism for divisible items that is IR, PO, and IC with 
public budgets. We also present a reduction to the setting with a single (and 
thus homogeneous) item that allows us to apply the following proposition from 
\cite{BhattacharyaEtAl10}. The basic building block of the mechanisms mentioned 
in this proposition is the ``adaptive clinching auction'' for a single divisible 
item. It is described for two agents in \cite{DobzinskiLaviNisan08}, as a 
``continuous time process'' for arbitrarily many agents in 
\cite{BhattacharyaEtAl10}, and as an explicit polynomial-time algorithm for 
arbitrarily many agents in Appendix~\ref{app:clinching-auction}.

\begin{proposition}[\cite{BhattacharyaEtAl10}]
For a single divisible item there exists a deterministic mechanism that 
satisfies IR, NPT, PO, and IC for public budgets. Additionally, for a single 
divisible or indivisible item there exists a randomized mechanism that satisfies 
IR ex interim, NPT ex post, PO ex post, and IC ex interim for private budgets.
\end{proposition}

\paragraph{Exact Characterizations of PO and IC.}
We start by characterizing PO outcomes through a simpler ``no trade'' 
condition. 
Outcome $(x,p)$ for single-dimensional valuations and either divisible or 
indivisible items that respects the budget limits satisfies {\em no trade 
(NT)} if 
(a) $\sum_{i \in N} x_{i,j} = 1$ for all $j \in M$, and
(b) there is no $x'$ such that for $\delta_i = \sum_{j \in M} (x'_{i,j} - 
    x_{i,j}) \alpha_j$ for all $i \in N$, $W = \{i \in N \mid \delta_i > 0 
    \}$, and $L = \{i \in N \mid \delta_i \le 0\}$ we have $\sum_{i \in N} 
    \delta_i v_i > 0$ and $\sum_{i \in W} \min(b_i - p_i, \delta_i v_i) + 
    \sum_{i \in L} \delta_i v_i \ge 0$.\footnote{For PO we only need that the 
outcome respects the {\em reported} budget limits. Hence our characterization 
also applies in {\em private} budget settings.}
This definition says that there should be no alternative assignment that 
overall increases the sum of the valuations, and allows the ``winners'' to 
compensate the ``losers''.
It differs from the definitions in prior work in that it allows trades that
involve both items {\em and} money. We will exploit this fact in the proof of our
impossibility result.
\begin{proposition}\label{prop:po-sing-dim}
Outcome $(x,p)$ for single-dimensional valuations and either divisible or 
indivisible items that respects the budget limits is PO if and 
only if it satisfies NT.
\end{proposition}

Next we characterize mechanisms that are IC with public budgets by ``value
monotonicity'' and ``payment identity''.
Mechanism $M=(x,p)$ for single-dimensional valuations and indivisible items 
that respects the publicly known budgets satisfies {\em value monotonicity (VM)} 
if for all $i \in N$, $\theta_i = (v_i,b_i)$, $\theta'_i = (v'_i,b_i)$, and 
$\theta_{-i} = (v_{-i}, b_{-i})$ we have that $v_i \le v'_i$ implies $\sum_{j 
\in M} x_{i,j}(\theta_i,\theta_{-i})\alpha_j \le \sum_{j \in M} x_{i,j}(\theta'_i,
\theta_{-i})\alpha_j.$ 
Mechanism $M=(x,p)$ for single-dimensional valuations and indivisible items 
that respects the publicly known budgets satisfies {\em payment identity (PI)} 
if for all $i \in N$ and $\theta = (v,b)$ with $c_{\gamma_t} \le v_i \le 
c_{\gamma_{t+1}}$ we have $p_i(\theta) = p_i((0,b_i), \theta_{-i}) + 
\sum_{s = 1}^{t} (\gamma_s -\gamma_{s-1}) c_{\gamma_s}(b_i, \theta_{-i})$, 
where $\gamma_0 < \gamma_1 < \dots$ are the values $\sum_{j \in M} x_{i,j}
\alpha_j$ can take and $c_{\gamma_s}(b_i,\theta_{-i})$ for $1 \le s \le t$ 
are the corresponding critical valuations.
While VM ensures that stating a higher valuation can only lead to a better
allocation, PI gives a formula for the payment in terms of the possible 
allocations and the critical valuations.
In the proof of our impossibility result we will use the fact that the payments 
for worse allocations provide a lower bound on the payments for better allocations.

\begin{proposition}\label{prop:ic-sing-dim}
Mechanism $M=(x,p)$ for single-dimensional valuations and indivisible items that 
respects the publicly known budgets is IC if and only if it satisfies VM and PI.
\end{proposition}

\paragraph{Deterministic Mechanisms for Indivisible Items.}
The proof of our impossibility result uses the characterizations of PO outcomes 
and mechanisms that are IC with public budgets as follows:
(a) PO is characterized by NT and NT induces a lower bound on the agents' 
    payments for a {\em specific} assignment, namely for the case that agent 
    1 only gets item $m$.
(b) IC, in turn, is characterized by VM and PI. Now VM and PI can be used to
    extend the lower bound on the payments for the {\em specific} assignment
    to {\em all} possible assignments.
(c) Finally, IR implies upper bounds on the payments that, with a suitable 
    choice of valuations, conflict with the lower bounds on the payments
    induced by NT, VM, and PI.

\begin{theorem}\label{thm:impossibility-sing-dim}
For single-dimensional valuations, indivisible items, and public budgets there
can be no deterministic mechanism $M=(x,p)$ that satisfies IR, PO, and IC.
\end{theorem}
\begin{proof}
For a contradiction suppose that there is a mechanism $M = (x,p)$ that is IR, 
PO, and IC for all $n$ and all $m$. Consider a setting with $n=2$ agents and 
$m=2$ items in which $v_1 > v_2 > 0$ and $b_1 > \alpha_{1} v_2.$

Observe that if agent~1's valuation was $v'_1 = 0$ and he reported his valuation
truthfully, then since $M$ satisfies IR his utility would be $u_1((0,b_1),
\theta_{-1},(0,b_1)) = - p_1((0,b_1), \theta_{-1}) \ge 0$. This shows that 
$p_1((0,b_1),\theta_{-1}) \le 0.$

By PO, which by \propref{prop:po-sing-dim} is characterized by NT, agent~1 with
valuation $v_1 > v_2$ and budget $b_1 > \alpha_{1}v_2$ must win at least one
item because otherwise he could buy any item from agent $2$ and compensate 
him for his loss.

PO, respectively NT, also implies that agent $1$'s payment for item $2$ must 
be strictly larger than $b_1 - (\alpha_{1} - \alpha_2)v_2$ because otherwise 
he could trade item $2$ against item $1$ and compensate agent~2 for his loss.

By IC, which by \propref{prop:ic-sing-dim} is characterized by VM and PI, 
agent $1$'s payment for item $2$ is given by $p_1(\{2\}) = 
p_1((0,b_1), \theta_{-1}) + \alpha_2 c_{\alpha_2}(b_1,\theta_{-1})$, where 
$c_{\alpha_2}$ is the critical valuation for winning item $2$. Together with 
$p_1(\{2\}) > b_1 -(\alpha_{1} - \alpha_2)v_2$ this shows that
$c_{\alpha_2}(b_1,\theta_{-1}) > (1/\alpha_2)[b_1-(\alpha_{1}-\alpha_2)v_2-
p_1((0,b_1),\theta_{-1})].$

IC, respectively VM and PI, also imply that agent~1's payment for any non-empty 
set of items $S$ in terms of the fractions $\gamma_t = \sum_{j \in S} \alpha_j 
> \dots > \gamma_1 = \alpha_2 > \gamma_0 = 0$ and corresponding critical 
valuations $c_{\gamma_t}(b_1,\theta_{-1}) \ge \dots \ge c_{\gamma_1} (b_1,
\theta_{-1}) = c_{\alpha_2}(b_1,\theta_{-1})$ is
$p_1(S) =   p_1((0,b_1),\theta_{-1}) + \sum_{s=1}^{t} (\gamma_s - \gamma_{s-1}) 
c_{\gamma_s}(b_1,\theta_{-1})$. Because $c_{\gamma_s}(b_1,\theta_{-1}) \ge 
c_{\alpha_2}(b_1,\theta_{-1})$ for all $s$ and $\sum_{s=1}^{t} (\gamma_s - 
\gamma_{s-1}) = \sum_{j \in S} \alpha_j$ we obtain $p_1(S) \ge p_1((0,b_1),
\theta_{-1}) + (\sum_{j \in S} \alpha_j) c_{\alpha_2}(b_1,\theta_{-1})$.

Combining this lower bound on $p_1(S)$ with the lower bound on $c_{\alpha_2}(b_1,
\theta_{-1})$ shows that $p_1(S) > (\sum_{j \in S}\alpha_j/\alpha_2) [b_1-(
\alpha_{1} - \alpha_2) v_2]$. 

For $v_1$ such that $(1/\alpha_2)[b_1-(\alpha_{1}-\alpha_2)v_2] > v_1 > v_2$ 
we know that agent~1 must win some item, but for any non-empty set of items $S$ 
the lower bound on agent~1's payment for $S$ contradicts IR.
\end{proof}

\paragraph{Randomized Mechanisms for Indivisible and Divisible Items.}
Interestingly, the impossibility result for deterministic mechanisms for 
indivisible items can be avoided by a randomized mechanism:
(a) Apply the randomized mechanism for a single {\em indivisible} item 
of \cite{BhattacharyaEtAl10} to a single indivisible item for which agent 
$i \in N$ has valuation $\vtilde_i = \sum_{j \in M} \alpha_j v_i$. 
(b) Map the single-item outcome $(\xtilde,\ptilde)$ into an outcome 
$(x,p)$ for the multi-item setting by setting $x_{i,j} = 1$ for all $j \in M$ 
if and only if $\xtilde_i = 1$ and setting $p_i = \ptilde_i$ for all $i \in N$. 

A similar idea works for divisible items. The only difference is that we use 
the mechanisms of \cite{BhattacharyaEtAl10} for a single {\em divisible} item, 
and map the single-item outcome $(\tilde{x}, \tilde{p})$ into a multi-item 
outcome by setting $x_{i,j} = \tilde{x}_i$ for all $i \in N$ and all $j \in M$ 
and setting $p_i = \tilde{p}_i$ for all $i \in N.$

The main difficulty in proving that the resulting mechanisms inherit the 
properties of the mechanisms in \cite{BhattacharyaEtAl10} is to show that
the resulting mechanisms satisfy PO (ex post). 
For this we argue that a certain structural property of the single-item 
outcomes is preserved by the mapping to the multi-item setting and remains 
to be sufficient for PO (ex post).

\begin{proposition}\label{prop:utilities}
Let $(\bar{x},\bar{p})$ be the outcome of our mechanism and let $(x,p)$ be the 
outcome of the respective mechanism of \cite{BhattacharyaEtAl10}, then 
$u_i(\bar{x}_i,\bar{p}_i) = u_i(x_i,p_i)$ for all $i \in N$ resp.~$E[u_i(\bar{x}_i,
\bar{p}_i)] = E[u_i(x_i,p_i)]$ for all $i \in N$. 
\end{proposition}

\begin{theorem}\label{thm:singdim-positive}
For single-dimensional valuations, divisible or indivisible items, 
and private budgets there is a randomized mechanism that satisfies IR ex 
interim, NPT ex post, PO ex post, and IC ex interim.
Additionally, for single-dimensional valuations and divisible items there 
is a deterministic mechanism that satisfies IR, NPT, PO, and IC for public 
budgets.
\end{theorem}
\begin{proof}
IR (ex interim) and IC (ex interim) follow from Proposition~\ref{prop:utilities} 
and the fact that the mechanisms of \cite{BhattacharyaEtAl10} are IR (ex interim)
and IC (ex interim). 
NPT (ex post) follows from the fact that the payments in our mechanisms and the 
mechanisms of \cite{BhattacharyaEtAl10} are the same, and the mechanisms 
in \cite{BhattacharyaEtAl10} satisfy NPT (ex post). 
For PO (ex post) we argue that the structural property of the outcomes of the
mechanisms in \cite{BhattacharyaEtAl10} that (a) $\sum_{i \in N} \xtilde_{i,j} 
= 1$ for all $j \in M$ and (b) $\sum_{j \in M} \xtilde_{i,j} > 0$ and 
$\vtilde_{i'} > \vtilde_i$ imply $\ptilde_{i'} = b_{i'}$ is preserved by the 
mapping to the multi-item setting and remains to be sufficient for PO (ex post).

We begin by showing that the structural property is preserved by the mapping. For this
observe that $\sum_{i \in N} \xtilde_{i,j} = 1$ for all $j \in M$ implies 
that $\sum_{i \in N} x_{i,j} = 1$ for all $j \in M$ and that  $\sum_{j \in M} 
\xtilde_{i,j} > 0$ and $\vtilde_{i'} > \vtilde_i$ imply $\ptilde_{i'} = b_{i'}$
implies that $\sum_{j \in M} x_{i,j} > 0$ and $v_{i'} > v_i$ imply $p_{i'} = b_{i'}$.

Next we show that the structural property remains to be sufficient for PO 
(ex post). For this assume by contradiction that the outcome $(x,p)$ 
is {\em not} PO (ex post). Then, by Proposition \ref{prop:po-sing-dim}, 
there exists an $x'$ such that $\sum_{i \in N} \delta_i v_i > 0$ and $\sum_{i 
\in W} \min(b_i-p_i,\delta_i v_i) + \sum_{i \in L} \delta_i v_i \ge 0$, where 
$\delta_i = \sum_{j \in M} (x'_{i,j} - x_{i,j}) \alpha_j$, $W = \{i \in N \mid 
\delta_i > 0\}$, and $L = \{i \in N \mid \delta_i \le 0\}$.

Because $(x,p)$ satisfies condition (a), i.e., $\sum_{i \in N} x_{i,j} = 1$ 
for all $j \in M$, and $x'$ is a valid assignment, i.e., $\sum_{i \in N} 
x'_{i,j} \le 1$ for all $j \in M$, we have $\sum_{i \in N} \delta_i = \sum_{j 
\in M} \sum_{i \in N} (x'_{i,j} - x_{i,j}) \alpha_j \le 0$. Because 
$\sum_{i \in N} \delta_i v_i > 0$ we have $\sum_{i \in W} \delta_i v_i 
\ge \sum_{i \in N} \delta_i v_i > 0$ and, thus, $\sum_{i \in W} \delta_i > 0$. 
We conclude that $\sum_{i \in L} \delta_i = \sum_{i \in N} \delta_i - \sum_{i 
\in W} \delta_i < 0$ and, thus, $\sum_{i \in L} \delta_i v_i < 0$. 

Because $(x,p)$ satisfies condition (b), i.e., $\sum_{j \in M} x_{i,j} > 0$ 
and $v_ {i'} > v_i$ imply $p_{i'} = b_{i'}$, there exists a $t$ with $1 \le 
t \le n$ such that (1) $\sum_{j \in M} x_{i,j} \ge 0$ and $p_{i} = b_{i}$ 
for $1 \le i \le t$, (2) $\sum_{j \in M} x_{i,j} \ge 0$ and $p_{i} \le b_{i}$ 
for $i = t+1$, and (3) $\sum_{j \in M} x_{i,j} = 0$ and $p_i \le b_i$ 
for $t+2 \le i \le n$.

{\em Case 1:} $t = n$. 
Then $\sum_{i \in W} \min(b_i-p_i, \delta_i v_i) =0$ and, thus, $\sum_{i \in W} 
\min(b_i-p_i,\delta_i v_i) + \sum_{i \in L} \delta_i v_i < 0$. 

{\em Case 2:} $t < n$ and $W \cap \{1,\dots,t\} = \emptyset$. 
Then $\sum_{i \in W} \delta_i v_i \le \sum_{i \in W} \delta_i v_{t+1}$ and
$\sum_{i \in L} \delta_i v_i \le \sum_{i \in L} \delta_i v_{t+1}$ and, thus,
$\sum_{i \in N} \delta_i v_i = \sum_{i \in W} \delta_i v_i + \sum_{i \in L} 
\delta_i v_i \le \sum_{i \in N} \delta_i v_{t+1} \le 0$. 

{\em Case 3:} $t < n$ and $W \cap \{1,\dots,t\} \neq \emptyset$. 
Then $\sum_{i \in W} \min(p_i-b_i,\delta_i v_i) \le \sum_{i \in W \setminus 
\{1..t\}} \delta_i v_{t+1}$ and $\sum_{i \in L} \delta_i v_i \le \sum_{i \in L} 
\delta_i v_{t+1}$ and, thus,  $\sum_{i \in W} \min(p_i-b_i,\delta_i v_i) + 
\sum_{i \in L} \delta_i v_i \le (\sum_{i \in N} \delta_i - \sum_{i \in W \cap 
\{1,\dots,t\}} \delta_i) v_{t+1} < 0$. 
\end{proof}


\section{Conclusion and Future Work}
In this paper we analyzed IR, PO, and IC mechanisms for settings with 
heterogeneous items. Our main accomplishments are: (a) An impossibility
result for {\em randomized} mechanisms and {\em public} budgets for 
additive valuations. (b) {\em Randomized} mechanisms that achieve these 
properties for {\em private} budgets and a restricted class of additive 
valuations. We are able to circumvent the impossibility result in the
restricted setting because our argument for the impossibility result is
based on the ability of an agent to overstate his valuation for one 
and understate his valuation for another item, which is not possible 
in the restricted setting. A promising direction for future work is to 
identify other valuations for which this is the case.


\bibliographystyle{abbrvnat}
\bibliography{abbshort,literature}


\appendix

\section{Proof of \propref{prop:char-multi-dim}}
Fix $i \in N$ and $\theta_{-i}=(v_{-i},b_{-i})$. By IC agent $i$ does not 
benefit from reporting $\theta'_i=(v'_i,b_i)$ when his true type is $\theta_i=
(v_i,b_i)$, nor does he benefit from reporting $\theta_i=(v_i,b_i)$ when his 
true type is $\theta'_i=(v'_i,b_i)$. Thus,
\begin{align*}
  v_i(x(\theta_i,\theta_{-i})) - p_i(\theta_i,\theta_{-i})    
    &\ge v_i(x(\theta'_i,\theta_{-i})) - p_i(\theta'_i,\theta_{-i})\\
  v'_i(x(\theta'_i,\theta_{-i})) - p_i(\theta'_i,\theta_{-i}) 
    &\ge v'_i(x(\theta_i,\theta_{-i})) - p_i(\theta_i,\theta_{-i})
\end{align*}
By combining these inequalities we get
\begin{align*}
  &v'_i(x_i(\theta'_i,\theta_{-i})) - v'_i(x_i(\theta_i,\theta_{-i})) 
     \ge v_i(x_i(\theta'_i,\theta_{-i})) - v_i(x_i(\theta_i,\theta_{-i})).
\end{align*}

\section{Analysis of Cases~1 and~2 in Section~\ref{sec:multidim}}\label{app:C1C2}

We start the analysis with an auxiliary lemma that shows that if 
at least one agent $i \in N$ has a positive valuation for some item $j \in
M$ then this item $j$ must be assigned completely in every outcome that 
is IR and \POall. 
\begin{lemma}\label{lem:everything-is-assigned}
If the valuation of at least one agent for an item $j\in M$ is positive, then an 
IR and \POall outcome assigns all of item $j$, i.e., $\sum_{i=1}^{n} x_{i,j}=1$.
\label{lem:assign-all}
\end{lemma}
\begin{proof}
Let us assume by contradiction that we have an outcome $(x,p)$ in which not all 
of the fractions of item~$j$ are assigned to the agents. Then the utility of the 
agents who have a positive valuation strictly increases when they get the unsold 
fractions of item~$j$ at price~0, while the utility of the other agents and that
of the auctioneer remain unchanged. Contradiction to \POall!
\end{proof}

Case~1 is easy: Agent~2 is not budget restricted and has the highest valuations 
for both items; so he will get both items. Thus in this case the utility for 
agent~1 is zero.
\begin{lemma}\label{lem:mult-case1}
Given \ica, \icb\ and \icc, then $x_{1,1}=0$, $x_{1,2}=0$, $x_{2,1}=1$, 
$x_{2,2}=1$, and $u_{1}=0$ in every IR and \POall outcome selected by an 
IC mechanism.
\end{lemma}
\begin{proof}
We divide the proof into the following parts: in \ref{it:mc1a} we show that $x_{1,1}=0$, $x_{1,2}=0$, $x_{2,1}=1$, and $x_{2,2}=1$, and in \ref{it:mc1b} we show that $u_{1}=0$.
\begin{enumerate}
\item
\label{it:mc1a}
Let us assume by contradiction that we have an IR and \POall outcome where 
$x_{1,1}>0$ or $x_{1,2}>0$. IR requires that $p_2 \leq x_{2,1} v_{2,1} + x_{2,2} 
v_{2,2}$. Hence, agent~2 can buy the fractions $x_{1,1}$ of item~1 and $x_{1,2}$ 
of item~2 for a payment $p$ with $x_{1,1} v_{2,1} + x_{1,2} v_{2,2}>p\geq x_{1,1} 
v_{1,1} + x_{1,2} v_{1,2}$ from agent~1. Because of \icb\ and \icc\ such a payment 
exists and agent~2 has enough money, since \ica\ implies
\begin{equation}
  b_2 > v_{2,1} + v_{2,2} 
    = (x_{1,1}+x_{2,1}) v_{2,1} + (x_{1,2}+x_{2,2}) v_{2,2}> p_2 + p.
\end{equation}
The utility of agent~2 would increase and the utilities of agent~1 and the 
auctioneer would not decrease. Contradiction to \POall!
\item
\label{it:mc1b}
We have already shown before that agent~1 gets no fraction of the items, and 
therefore, IR implies that his payments cannot be positive.

Let us consider the subcase where $v_{1,1}=v_{1,2}=0$ and agent~1 reports 
truthfully. The valuations of agent~2 are positive. Because of IR the payment 
of agent~2 cannot exceed his reported valuation, but \ref{it:mc1a} holds when 
his reported valuations are positive. Therefore, agent~2 would have an incentive
to understate his valuation when his payment would be positive. Hence, IR of the 
auctioneer implies that the payment of both agents is equal to~0. This means, 
that the utility of agent~1 is~0 in this case.

If there would exist any other reported valuation of agent~1, where he gets no 
items, but where his payments are negative, then he would have an incentive to 
lie, when his valuations are equal to~0. This would contradict IC!
\end{enumerate}
\end{proof}

In Case~2, agent~1 has the higher valuation for item~1, while agent~2 has the 
higher valuation for item~2. Thus, agent~1 gets item~1 and agent~2 gets item~2. 
Since the only difference to Case~1 is that in Case~2 $v_{1,1} > v_{2,1}$ while 
in Case~1 $v_{1,1} < v_{2,1}$, the critical value whether agent~2 gets item~1 or 
not is $v_{2,1}$, and thus in every IC mechanism, agent~1 has to pay $v_{2,1}$ 
and his utility is $v_{1,1} - v_{2,1}$.
\begin{lemma}\label{lem:mult-case2}
Given \ica, \icbp, \icc, and \icd, then $x_{1,1}=1$, $x_{1,2}=0$, 
$x_{2,1}=0$, $x_{2,2}=1$, and $u_1=v_{1,1}-v_{2,1}$ in every IR and \POall 
outcome selected by an IC mechanism. 
\end{lemma}
\begin{proof}
We divide the proof into the following parts: in \ref{it:mc2a} we show that 
$x_{1,1}=1$, $x_{1,2}=0$, $x_{2,1}=0$, and $x_{2,2}=1$, and in \ref{it:mc2b} 
we show that $u_1=v_{1,1}-v_{2,1}$.
\begin{enumerate}
\item \label{it:mc2a}
Let us assume by contradiction that $x_{1,2}>0$. Then, agent~2 can buy these 
fractions of item~2 for a payment $p$ with $x_{1,2} v_{2,2} > p \geq x_{1,2} 
v_{1,2}$, which exists because of \icc. IR  and \ica\ ensure that agent~2 has 
enough budget, since $b_2 > v_{2,1} + v_{2,2} = (x_{1,1}+x_{2,1}) v_{2,1} + 
(x_{1,2}+x_{2,2}) v_{2,2} \geq p_2 + x_{1,1} v_{2,1} + x_{1,2} v_{2,2} > p_2 + 
p$. The utility of the agent~2 would increase, while the utilities of agent~1 
and the auctioneer would not decrease. Contradiction to \POall!

Otherwise, let us assume that $x_{1,1}<1$ and $x_{1,2}=0$. Then, agent~1 can 
buy the other fractions of item~1 for a payment $p$ with $x_{2,1} v_{1,1} > p 
\geq x_{2,1} v_{2,1}$, which exists because of \icbp. IR and \icd\ 
ensure that agent~1 has enough budget, since $b_1 > v_{1,1} = (x_{1,1}+x_{2,1}) 
v_{1,1} \geq p_1 + x_{2,1} v_{1,1} > p_1 + p$. The utility of agent~1 would 
increase, while the utilities of agent~2 and the auctioneer would not decrease. 
Contradiction to \POall!
\item\label{it:mc2b}
We show first that $p_1\leq v_{2,1}$. Since $x_{1,1}=1$ and $x_{1,2}=0$, IR 
requires that $p_1 \leq v_{1,1}$. If $p_1>v_{2,1}$, then agent~1 has an 
incentive to lie. If he states that his valuation for item~1 is $v'_{1,1}$ with 
$p_1>v'_{1,1}>v_{2,1}$, then the allocation of the items does not change, but he 
pays less because of IR. Contradiction to IC!

Now, we show that $p_1\geq v_{2,1}$. Let us therefore assume by contradiction 
that $p_1 < v_{2,1}$. If we have $v'_{1,1}$ with $p_1<v'_{1,1}<v_{2,1}$ instead 
of $v_{1,1}$, and all the other valuations are left unchanged, then 
Lemma~\ref{lem:mult-case1} implies that $u'_1=0$. Hence, in this case agent~1 
can increase his utility when he lies and states that his valuation is $v_{1,1}$, 
because his utility would be $v'_{1,1}-p_1>0$. Contradiction to IC!

Since agent~1 gets all fractions of item~1, no fraction of item~2, and has to 
pay $v_{2,1}$, his utility is $v_{1,1} - v_{2,1}$.
\end{enumerate}
\end{proof}

\section{Proof of Lemma~\ref{lem:mult-case3a}}
We divide the proof into the following parts: in \ref{it:l1a} we show that 
$x_{1,1}=1$ and $x_{1,2}=1$ if $p_{1}<b_{1}$, in \ref{it:mc3a} we show that 
$x_{1,2}>(1-x_{1,1}) \frac{v_{2,1}}{v_{2,2}}$ if $p_{1}=b_{1}$, and in 
\ref{it:l1c} we show that $x_{1,1}=1$ and $x_{1,2}>0$ if $p_{1}=b_{1}$.
\begin{enumerate}
\item\label{it:l1a}
Let us assume by contradiction that $p_1<b_1$ and $x_{1,j}<1$ for an item $j\in
\{ 1,2\}$. Agent~1 can increase his utility by buying $\min\{\frac{b_1-p_1}{p},
x_{2,j}\}$ fractions of item $j$ for a unit price $p$ with $v_{1,j} > p \geq 
v_{2,j}$ from agent~2. Such a price exists, because of \icbp\ and 
\iccp. Agent~1 has enough money for the trade, since $p_1 + p 
\min\{\frac{b_1-p_1}{p},x_{2,j}\} = \min\{ b_1, p_1 + p x_{2,j}\} \leq b_1$. 
The utility of agent~1 would increase, and the utilities of agent~2 and the 
auctioneer would not decrease. Contradiction to \POall!
\item
\label{it:mc3a}
IR requires $b_1 = p_1\leq v_{1,1} x_{1,1} + v_{1,2} x_{1,2}$, and therefore, 
$x_{1,2}\geq \frac{b_1-v_{1,1}x_{1,1}}{v_{1,2}}$. If $x_{1,1}=1$, then \icd\ 
implies that $(1-x_{1,1}) \frac{v_{2,1}}{v_{2,2}}=0<\frac{b_1-v_{1,1}}{v_{1,2}}
=\frac{b_1-v_{1,1}x_{1,1}}{v_{1,2}}$. Otherwise, if $x_{1,1}=0$, then \icd\ and 
\ice\ imply that $(1-x_{1,1}) \frac{v_{2,1}}{v_{2,2}}=\frac{v_{2,1}}{v_{2,2}}<
\frac{b_1}{v_{1,2}}=\frac{b_1-v_{1,1}x_{1,1}}{v_{1,2}}$, and hence, $(1-x_{1,1}) 
\frac{v_{2,1}}{v_{2,2}}<\frac{b_1-v_{1,1}x_{1,1}}{v_{1,2}}$ for all $x_{1,1}\in 
[0,1]$. Therefore, we have that $(1-x_{1,1}) \frac{v_{2,1}}{v_{2,2}}<x_{1,2}$ 
for all possible values of $x_{1,1}$.
\item\label{it:l1c}
We split the proof into two parts. We assume by contradiction that 
either $p_1=b_1$, $x_{1,1}\leq 1$ and $x_{1,2}=0$, or that $p_1=b_1$, 
$x_{1,1}<1$ and $x_{1,2}>0$.

Let us assume that $p_1=b_1$, $x_{1,1}\leq 1$ and $x_{1,2}=0$. According to 
\icd, the 
utility of agent~1 is negative. Contradiction to IR!

We will now investigate the other case and assume that $p_1=b_1$, $x_{1,1}<1$ 
and $x_{1,2}>0$. Agent~2 has the same valuation for $x_{1,2}=1-x_{1,1}$ 
fractions of item~1 and $(1-x_{1,1}) \frac{v_{2,1}}{v_{2,2}}$ fractions of 
item~2. The valuation of agent~1 for $(1-x_{1,1}) \frac{v_{2,1}}{v_{2,2}}$ 
fractions of item~2 is identical to the valuation for $(1-x_{1,1}) \frac{v_{2,1} 
v_{1,2}}{v_{2,2} v_{1,1}}$ fractions of item~1. We know that ${v_{2,1} v_{1,2}}
<{v_{2,2} v_{1,1}}$. That is, that the utility of agent~1 is 
increased and the utilities of agent~2 and the auctioneer are not decreased, 
when agent~1 trades $(1-x_{1,1}) \frac{v_{2,1}}{v_{2,2}}$ fractions of item~2 
against $x_{2,1}=1-x_{1,1}$ fractions of item~1. Fact \ref{it:mc3a} implies that 
agent~1 actually has the required $(1-x_{1,1}) \frac{v_{2,1}}{v_{2,2}}$ 
fractions of item~2. Contradiction to \POall! 
\end{enumerate}

\section{Proof of Lemma~\ref{lem:mult-case3b}}
We divide the proof into the following parts: in \ref{it:mc4a} we show that 
$p_{1}=b_{1}$ and $x_{1,2}<1$, in \ref{it:mc4b} we show that 
$\frac{b_1-v_{2,1}}{v_{2,2}} \geq x_{1,2} \geq \frac{b_{1}-v_{2,1}}{v_{1,2}}$, 
and in  \ref{it:mc4cn} we show that $x_{1,2}=\frac{b_1 - v_{2,1}}{v_{2,2}}$.
\begin{enumerate}
\item
\label{it:mc4a}
Lemma~\ref{lem:mult-case3a} implies that the utility of agent~1 is $v_{1,1} + 
x_{1,2} v_{1,2} - p_1$. We know that $v_{2,1}+v_{2,2} > b_1$. 
Hence, we can select a sufficiently small $\epsilon > 0$ such that $v_{2,1} +
v_{2,2} -\epsilon > b_1$. Because of \icbp\ and \icd, we know that 
$v_{2,2}-\epsilon > 0$. Let us consider the case where we have $v'_{1,2}$ with 
$v_{2,2}>v'_{1,2}>v_{2,2}-\epsilon$ instead of $v_{1,2}$ and all other valuation 
are left unchanged. In this case, the utility of agent~1 is $v_{1,1} - v_{2,1}$, 
because of Lemma~\ref{lem:mult-case2} and since $v_{2,2}>v'_{1,2}$ holds. 
Therefore, IC implies that
\begin{equation}
  v_{1,1} - v_{2,1} \geq v_{1,1} + x_{1,2} v'_{1,2} - p_1.
\label{eq:mc4a}
\end{equation}
Let us assume by contradiction that $x_{1,2}=1$, then equation (\ref{eq:mc4a}) 
implies
\begin{equation}
  p_1 \geq v_{2,1} + v'_{1,2} > v_{2,1} + v_{2,2} - \epsilon > b_1,
\end{equation}
which contradicts the budget constraint. Therefore, $x_{1,2}<1$, 
and hence, Lemma~\ref{lem:mult-case3a} implies that $p_1=b_1$.
\item
\label{it:mc4b}
Lemma~\ref{lem:mult-case3a} and \ref{it:mc4a} show that the utility of agent~1 
is $v_{1,1} + x_{1,2} v_{1,2} - b_1$. We select a sufficiently small $\epsilon 
> 0$, such that $v_{2,1} +v_{2,2} -\epsilon > b_1$ and consider the case where 
$v'_{1,2}=v_{2,2}-\epsilon$ and all other valuations are unchanged. 
Lemma~\ref{lem:mult-case2} implies that the utility of agent~1 is $v_{1,1}-
v_{2,1}$ in this case. Hence, IC implies that
\begin{equation}
  v_{1,1} - v_{2,1} \geq v_{1,1} + x_{1,2} v'_{1,2} - b_1
\label{eq:mc4b1}
\end{equation}
and
\begin{equation}
  v_{1,1} + x_{1,2} v_{1,2} - b_1 \geq v_{1,1} -v_{2,1}.
\label{eq:mc4b2}
\end{equation}
Inequality (\ref{eq:mc4b1}) implies that $\frac{b_1-v_{2,1}}{v_{2,2}-\epsilon}=
\frac{b_1-v_{2,1}}{v'_{1,2}}\geq x_{1,2}$. Since this inequality has to hold for 
all sufficiently small $\epsilon >0$, we know that $\frac{b_1 - v_{2,1}}{v_{2,2}}
\geq x_{1,2}$. Inequality (\ref{eq:mc4b2}) implies that $\frac{b_1-v_{2,1}}{v_{1,2}}
\leq x_{1,2}$.
\item\label{it:mc4cn}
Let us assume by contradiction that the inequality $\frac{b_1-v_{2,1}}{v_{2,2}} 
\geq x_{1,2}$ implied by \ref{it:mc4b} is strict, and $\gamma > 0$ is defined 
such that $\frac{b_1-v_{2,1}}{v_{2,2}} = x_{1,2} + \gamma$. We select arbitrary 
$\epsilon>0$ and $\delta$ with $ v_{2,2} \left( \frac{b_1 -v_{2,1}}{b_1 - 
v_{2,1} - \gamma v_{2,2}} -1 \right) >\delta>0$ which fulfill $v_{1,2}-\epsilon 
-\delta =v_{2,2}$. Such variables $\epsilon$ and $\delta$ exist because of 
\iccp, and since \icbp, \icd\ and $\gamma>0$ imply 
that $\frac{b_1 -v_{2,1}}{b_1 - v_{2,1} - \gamma v_{2,2}} > 1$. We consider the 
alternative case where $v'_{1,2}=v_{1,2}-\epsilon$ and all other valuations are 
unchanged. In this case, \ref{it:mc4b} implies that $\frac{b_1-v_{2,1}}{v'_{1,2}}
\leq x'_{1,2}$, and hence, $\frac{b_1-v_{2,1}}{v_{2,2} + \delta}\leq x'_{1,2}$. 
Furthermore, Lemma~\ref{lem:mult-case3a} and \ref{it:mc4a} imply that $p_1=b_1$ 
and $x_{1,1}=1$ in both cases. Now, IC requires that
\begin{equation}
  v_{1,1} + x_{1,2} v_{1,2} - b_1 \geq v_{1,1} + x'_{1,2} v_{1,2} -b_1,
\end{equation}
respectively $x_{1,2}\geq x'_{1,2}$, and therefore, $\frac{b_1 -v_{2,1}}{v_{2,2}} 
- \gamma \geq \frac{b_1 - v_{2,1}}{v_{2,2} + \delta}$. But this inequality can 
be transformed to $\delta \geq v_{2,2} \left( \frac{b_1 -v_{2,1}}{b_1 - v_{2,1} 
- \gamma v_{2,2}} -1 \right)$. Contradiction!
\end{enumerate}

\section{Adaptive Clinching Auction for a Single Divisible Item}
\label{app:clinching-auction}

We investigate the adaptive clinching auction for a single divisible item that 
is described as a ``continuous time process'' in \cite{BhattacharyaEtAl10} in 
order to construct an explicit algorithm. In Step (II) of the differential 
process described in \cite{BhattacharyaEtAl10} the item is overdemanded and no 
bidder exits the auction because his valuation is identical to the current price. 
This is the case which has to be analyzed. We consider a time span $[t_1,t_2]$. 
$A$ is the set of active bidders at time $t_1$ and $C$ is the set of clinching 
bidders at time $t_1$. For all $t \in [t_1,t_2]$ we have $p(t)=p(t_1)+(t-t_1)$. 
We assume that $t_2$ is selected such that $v_i > p(t)$ for all $i \in A$ 
and $t \in [t_1,t_2)$. Therefore, the set of active bidders at time $t\in [t_1,
t_2)$ is equal to $A$ and the set of exiting bidders at time $t\in [t_1,t_2)$ is 
empty. We assume further that $t_2$ is selected such that no bidder starts 
clinching during $(t_1,t_2)$, and that the demand $D(t)=\sum_{i\in A}\frac{b_i(t)}
{p(t)}$ is larger than the supply $S(t)$ for all $t\in [t_1,t_2)$. Hence, at every 
time $t$ in $(t_1,t_2)$ Step (II) of the process is selected and the set of 
clinching bidders $C$ does not change.

Consider time $t$ in $(t_1,t_2)$. By the definition of the clinching bidders 
the supply is given by $S(t)=\sum_{j\in A\setminus\{ i\}} \frac{b_j(t)}{p(t)}$ 
for all $t \in (t_1,t_2)$ and every clinching bidder $i \in C$. Since every 
clinching bidder $i \in C$ gets the same fraction allocated during $(t_1,t)$ we 
have
\begin{equation*}
x_i(t)-x_i(t_1) = \frac{S(t_1)-S(t)}{|C|}
                = \frac{1}{|C|}\sum_{j\in A\setminus \{ i\}}\left( \frac{b_j(t_1)}
                  {p(t_1)} - \frac{b_j(t)}{p(t)} \right).
\end{equation*}
Let us now differentiate this equation with respect to $t$. We get 
\begin{equation*}
x'_i(t) = \frac{1}{|C|}\sum_{j\in A\setminus \{ i\}}\left( - \frac{b'_j(t)p(t) - b_j(t)p'(t)}{p(t)^2} \right)
        = \frac{1}{|C|}\sum_{j\in A\setminus \{ i\}}\left( \frac{b_j(t)}{p(t)^2} - \frac{b'_j(t)}{p(t)} \right).
\end{equation*}
Bidder $i \in C$ pays $p(t)$ for the fractions that he is clinching at time $t$. 
Hence $b'_i(t) = - x'_i(t) p(t)$. This, the previous equality, $b'_j(t) = 0$ for 
$j \in A \setminus C$, and $b'_j(t) = b'_i(t)$ and $b_j(t) = b_i(t)$ for $j \in C 
\setminus \{i\}$ implies
\begin{equation*}
b'_i(t) = -\sum_{j\in A\setminus \{ i\}} \frac{b_j(t)}{p(t)}
        = \frac{-\sum_{j\in A\setminus C}{b_j(t_1)}-(|C|-1)b_i(t)}{p(t_1)-t_1+t}.
\end{equation*}
For the case that $|C|>1$ we can solve this differential equation and obtain
\begin{equation*}
b_i(t)=\frac{1}{|C|-1}\left(\left( \frac{p(t_1)}{p(t)} \right)^{|C|-1} 
       \sum_{j\in A\setminus \{ i \}} b_j(t_1) -\sum_{j\in A\setminus C} b_j(t_1)\right).
\end{equation*}
Since $b_j(t)=b_j(t_1)$ for all $j\in A\setminus C$ and $b_j(t)=b_i(t)$ for all 
$j\in C$ it follows that
\begin{equation}
\sqrt[|C|-1]{\frac{\sum_{j\in A\setminus\{ i\}} b_j(t_1)}{\sum_{j\in A\setminus\{ i\}} b_j(t)}}
=\frac{p(t)}{p(t_1)}.
\label{eq:priceCg1}
\end{equation}
For the case that $|C|=1$ we have
\begin{equation*}
b'_i(t)= -\frac{\sum_{j\in A\setminus \{i\}} b_j(t_1)}{p(t_1)-t_1+t}
\end{equation*}
and obtain
\begin{equation}
\exp\left( \frac{b_i(t_1)-b_i(t)}{\sum_{j\in A\setminus \{ i\}} b_j(t_1)}\right) 
= \frac{p(t)}{p(t_1)}.
\label{eq:priceCis1}
\end{equation}
Equations (\ref{eq:priceCg1}) and (\ref{eq:priceCis1}) allow us to compute the 
prices where a new bidder would start clinching.

The construction of the algorithm follows the differential process described in 
\cite{BhattacharyaEtAl10}. Lines 28-35 of Algorithm~\ref{alg:clinch} correspond 
to Step (I) in \cite{BhattacharyaEtAl10}; lines 13-23 of Algorithm~\ref{alg:clinch} 
correspond to Step (III) in \cite{BhattacharyaEtAl10}; and Algorithm~\ref{alg:contclinch}, 
which is called on line~10 of Algorithm~\ref{alg:clinch} corresponds to Step (II) 
in \cite{BhattacharyaEtAl10}. Line 16, 20, and 22 of Algorithm~\ref{alg:contclinch} 
follow from equation (\ref{eq:priceCg1}) and (\ref{eq:priceCis1}). The variables 
that are used in Algorithm~\ref{alg:clinch} and Algorithm~\ref{alg:contclinch} 
are described in Table~\ref{tab:vars}.

For the running time observe that each time one of the two while-loops gets 
executed either an active bidder who was not clinching becomes a clinching 
bidder or an active bidder becomes an exiting bidder. Since an exiting bidder 
cannot become active again and an active clinching bidder cannot become an 
active non-clinching bidder again it follows that the algorithm runs in time 
polynomial in the bidders.

\begin{table}[h]
\caption{Description of the Variables in Algorithm~\ref{alg:clinch}}{
\begin{tabular}{clcl}
\hline
\textit{Variable}&\textit{Data Type}&\textit{Constraint}&\textit{Description}\\\hline
$n$&integer (constant)&$n>1$&number of bidders\\
$b$&real vector (length $n$)&$b_i>0\ \forall i\in \{ 1,\dots,n\}$&budgets\\
$v$&real vector (length $n$)&$v_i>0\ \forall i\in \{ 1,\dots,n\}$&valuations\\
$A$&set of integers&$A \subseteq\{ 1,\dots,n\}$&set of active bidders\\
$E$&set of integers&$E \subseteq\{ 1,\dots,n\}$&set of exiting bidders\\
$C$&set of integers&$C \subseteq\{ 1,\dots,n\}$&set of clinching bidders\\
$p$&real&$p\geq 0$&price\\
$x$&real vector (length $n$)&$x_i\geq 0\ \forall i\in \{ 1,\dots,n\}$&allocated amount\\
$S$&real&$S\geq 0$&supply\\
$D$&real&$D\geq 0$&aggregated demand\\
\hline
\end{tabular}
}
\label{tab:vars}
\end{table}

\begin{algorithm}[H]
\caption{Adaptive Clinching Auction for a Single Divisible Good.}
\label{alg:clinch}
\begin{algorithmic}[1]
\small
\Procedure{Clinching}{$n,b,v$}
	\State \Comment{initialize variables}
	\State $(p,S,D)\gets (0,1,\infty)$
	\State $(A,E,C)\gets (\{1,\dots,n\},\emptyset,\emptyset)$
        \State $x_i\gets 0\ \forall i\in A$
	\While{$D> S$} 
	\algstore{alg:slot:1}
	\end{algorithmic}
	\end{algorithm}
	\addtocounter{algorithm}{-1}
        	
	\begin{algorithm}[H]
	\caption{Adaptive Clinching Auction for a Single Divisible Good \cont.}
	\begin{algorithmic}[1]
	\small
	\algrestore{alg:slot:1}
        \State \Comment{item is overdemanded}
		\If{$E=\emptyset$}
			\State \Comment{there are no exiting bidders}
			\State $(A,E,C,p,S,D,x,b)\gets\textsc{ContinuousClinching}(A,C,p,S,D,x,b,v)$
		\Else
			\State \Comment{there are exiting bidders}
			\State $m\gets \sum_{i\in E}b_i$
			\State $E\gets\emptyset$
			\While{$m>0$}
				\State\Comment{compute amount that the clinching bidders can clinch}
				\State\Comment{before a new bidder starts clinching}
				\State $c\gets \min\{|C|(\min_{i\in C} b_i - 
                                       \max_{i\in A\setminus C} b_i), m\}$
				\State $m\gets m-c$
				\State $(x_i,b_i)\gets (x_i + \frac{c}{|C|p}, 
                                       b_i-\frac{c}{|C|})\ \forall i\in C$
				\State $(S,D)\gets (S-\frac{c}{p},D-\frac{c}{p})$
 				\State $C\gets\{ i\in A| D-\frac{b_i}{p}=S\}$
			\EndWhile
		\EndIf
	\EndWhile
	\State \Comment{item is not overdemanded}
	\State \Comment{sell to active bidders the amount they can afford}
	\State $S\gets S-D$
	\State $(x_i,b_i)\gets (x_i +\frac{b_i}{p},0)\ \forall i\in A$
	\State \Comment{sell left fractions to exiting bidders}
	\For{$i\in E$}
		\State $m\gets \min\{ \frac{b_i}{p}, S\}$
		\State $(x_i,b_i)\gets (x_i + m, b_i - m p)$
		\State $S\gets S - m$
	\EndFor
	\State \textbf{return} $(x,b)$
\EndProcedure
\end{algorithmic}
\end{algorithm}

\vspace{-12pt}

\begin{algorithm}[H]
\caption{Continuous Clinching.}
\label{alg:contclinch}
\begin{algorithmic}[1]
\small
\Procedure{ContinuousClinching}{$A,C,p,S,D,x,b,v$}
    \If{$C=\emptyset$}
       	\State \Comment{compute highest price $p$ where no bidder clinched}
       	\State \Comment{or exited the auction before}
 	\State $p^*\gets \frac{\sum_{i\in A}b_i - \max_{i\in A}b_i}{S}$
	\State $p\gets \min \{ p^*,\min_{i\in A} v_i\}$
        \State\Comment{update variables}
	\State $(A,E)\gets(\{i \in A| v_i > p\},\{i \in A| v_i = p\})$
	\algstore{alg:slot:1}
	\end{algorithmic}
	\end{algorithm}
	\addtocounter{algorithm}{-1}
        	
	\begin{algorithm}[H]
	\caption{Continuous Clinching \cont.}
	\begin{algorithmic}[1]
	\small
	\algrestore{alg:slot:1}
	\State $D\gets \frac{\sum_{i\in A} b_i}{p}$
	\State $C\gets \{i\in A| D-\frac{b_i}{p}=S\}$
    \Else
    	\State \Comment{compute next break point of the differential process}
    	\State \Comment{and update variables}  
        \State $b^*\gets \max_{j \in A\setminus C} b_j$
        \State\Comment{price where a new bidder would start to clinch}
	\State $p^*\gets 
               \begin{cases}
	       p \exp \left( \frac{\max_{i\in A b_i - b^*}}{\sum_{i\in A
                 \setminus C}b_i}\right),\ &\text{if}\ |C|=1\\
	       p \left( \sqrt[|C|-1]{\frac{p S}{(|C|-1)b^*+\sum_{i\in A
                 \setminus C}b_i}}\right),\ &\text{if}\ |C|>1
	       \end{cases}$
	\State \Comment{price at the next break point}
	\State $\tilde{p}\gets \min\{ p^*,\min_{i\in A} v_i\}$ 
        \State \Comment{supply at the next break point}
	\State $\tilde{S}\gets (\frac{p}{\tilde{p}})^{|C|}S$
	\State \Comment{budget of the clinching bidder at next break point} 
%
	\State $\tilde{b}\gets 
               \begin{cases}
	       \max_{i\in A}b_i - \log\left(\frac{\tilde{p}}{p}\right)
               \sum_{i\in A\setminus C}b_i,\ &\text{if}\ |C|=1\\
		  \frac{1}{|C|-1}\left(  S \frac{p^{|C|}}{\tilde{p}^{|C|-1}} -
               \sum_{i\in A\setminus C} b_i \right),\ &\text{if}\ |C|>1
	       \end{cases}$
	\State\Comment{update variables}
	\State $(x_i,b_i)\gets (x_i + \frac{1}{|C|} (S-\tilde{S}),\tilde{b})\ \forall i\in C$
	\State $E\gets\{i \in A| v_i = p\}$
	\State $A\gets\{i \in A| v_i > p\}$
	\State $C\gets \arg \max_{i\in A} b_i$
	\State $(p,S,D)\gets (\tilde{p},\tilde{S},\sum_{i\in A}\frac{b_i}{p})$
    \EndIf
    \State \textbf{return} $(A,E,C,p,S,D,x,b)$
\EndProcedure
\end{algorithmic}
\end{algorithm}

\section{Proof of \propref{prop:po-sing-dim}}
First we show that if $(x,p)$ satisfies PO, then it satisfies NT. To this end we
show that if $(x,p)$ does {\em not} satisfy NT, then it is {\em not} PO.

{\em Case~1:} $\lnot$ NT because $\lnot$ (a)

There exists an item $j \in M$ such that $\sum_{i \in N} x_{i,j} < 1$. By 
assumption every agent $i \in N$ has $v_i > 0$ and, thus, $\alpha_j v_i > 
0.$ Consider the outcome $(x',p')$ that results from assigning the unassigned
fraction of item $j$ to some agent $i' \in N$ at no additional cost. For this 
outcome we have $u'_i = u_i$ for all agents $i \in N \setminus \{i'\}$, $u'_{i'} 
> u_{i'}$ for agent $i'$, and $\sum_{i \in N} p'_i = \sum_{i \in N} p_i$. Hence
$(x,p)$ is {\em not} PO.

{\em Case~2:} $\lnot$ NT because $\lnot$ (b)

There exists an assignment $x'$ such that $\sum_{i \in N} \delta_i v_i > 0$ and 
$\sum_{i \in W} \min(b_i - p_i, \delta_i v_i) + \sum_{i \in L} \delta_i v_i \ge 
0$. Consider the outcome $(x',p')$ for which $p'_i = p_i + \min(b_i - p_i, 
\delta_i v_i)$ for all agents $i \in W$ and $p'_i = p_i + \delta_i v_i$ for all 
agents $i \in L$.

For all agents $i \in N$ we have $u'_i \ge u_i$ because
\begin{align}
  u'_i &= \sum_{j \in M} x'_{i,j} \alpha_j v_i - p'_i \notag\\
       &= \sum_{j \in M} x_{i,j} \alpha_j v_i + \delta_i v_i - p_i - \min(b_i - 
          p_i, \delta_i v_i) \notag\\
       &\ge u_i, &&\text{for $i \in W$, and}\label{eq:agent} \displaybreak[0]\\
  u'_i &= \sum_{j \in M} x'_{i,j} \alpha_j v_i - p'_i \notag\\  
       &= \sum_{j \in M} x_{i,j} \alpha_j v_i+ \delta_i v_i - p_i - \delta_i v_i \notag\\
       &= u_i &&\text{for $i \in L$.}\notag
\end{align}

For the auctioneer we have $\sum_{i \in N} p'_i \ge \sum_{i \in N} p_i$ because
\begin{align}
  \sum_{i \in N}p'_i - \sum_{i \in N}p_i
  &= \sum_{i \in W}p'_i + \sum_{i \in L}p'_i - \sum_{i \in N}p_i \displaybreak[0] \notag\\
  &= \sum_{i \in W}(p_i + \min(b_i-p_i, \delta_i v_i)) + \sum_{i \in L} (p_i + 
     \delta_i v_i) - \sum_{i \in N}p_i \displaybreak[0] \notag\\
  &= \sum_{i \in W}\min(b_i - p_i, \delta_i v_i) + \sum_{i \in L}\delta_i v_i \notag\\
  &\ge 0. \label{eq:auctioneer}         
\end{align}

If $\sum_{i \in W} \min(b_i - p_i, \delta_i v_i) + \sum_{i \in L} \delta_i v_i > 
0$, then inequality (\ref{eq:auctioneer}) is strict showing that $\sum_{i \in 
N} p'_i > \sum_{i \in N} p_i$. Otherwise, $\sum_{i \in W} \min(b_i - p_i, 
\delta_i v_i) + \sum_{i \in L} \delta_i v_i = 0$, and since $\sum_{i \in N} 
\delta_i v_i > 0$ we must have $b_i - p_i < \delta_i v_i$ for at least one agent 
$i\in W$. For this agent $i$ inequality (\ref{eq:agent}) is strict showing that
$u'_i > u_i.$ Hence in both cases $(x,p)$ is {\em not} PO.

Next we show that if $(x,p)$ satisfies NT, then it is PO. To this end we show 
that if $(x,p)$ is {\em not} PO, then it does {\em not} satisfy NT. If $(x,p)$
is {\em not} PO, then there exists an outcome $(x',p')$ such that $u'_i \ge u_i$ 
for all agents $i \in N$ and $\sum_i p'_i \ge \sum_i p_i$, with at least one of 
the inequalities strict.

If {\em not} all items are assigned completely in $(x,p)$, then we have $\lnot$ 
(a) and so $(x,p)$ does {\em not} satisfy NT. Otherwise, if in $(x,p)$ all items 
are assigned completely, then to show that $(x,p)$ does {\em not} satisfy NT we 
have to show $\lnot$ (b). To this end consider the assignment $x'$ and let 
$\delta_i = \sum_{j \in M} (x'_{i,j} - x_{i,j}) \alpha_j$ for $i \in N$, let $W 
= \{i \in N \mid \delta_i > 0\}$, and let $L = \{i \in N \mid \delta_i \le 0\}.$

We begin by showing that $\sum_{i \in W} \min(b_i - p_i, \delta_i v_i) + \sum_{i 
\in L} \delta_i v_i \ge 0.$ 

For $i \in N$ we have $p'_i - p_i \le \min(b_i - p_i, \delta_i v_i)$ because
\begin{align}
  &p'_i \le b_i \ \ \Rightarrow \ p'_i - p_i \le b_i - p_i, && \text{and} \notag \\
  &u'_i \ge u_i \ \Rightarrow \ p'_i - p_i \le \delta_i v_i. \notag
\end{align}

It follows that
\begin{align}
  \sum_{i \in W} \min(b_i - p_i, \delta_i v_i) + \sum_{i \in L} \delta_i v_i 
  &\ge \sum_{i \in W} (p'_i - p_i) + \sum_{i \in L} (p'_i - p_i) \notag\\
  &= \sum_{i \in N} p'_i - \sum_{i \in N} p_i
  \ge 0. \notag
\end{align}

Next we show that $\sum_{i \in N} \delta_i v_i > 0$.

Since $u'_i \ge u_i$ for all $i \in N$ and $\sum_{i \in N} p'_i \ge \sum_{i \in 
N} p_i$ we have
\begin{align}
  \sum_{i \in N} u'_i \ge \sum_{i \in N} u_i
  &\Leftrightarrow \sum_{i \in N} (\sum_{j \in M} x'_{i,j} \alpha_j v_i - p'_i) 
                   \ge 
                   \sum_{i \in N} (\sum_{j \in M} x_{i,j} \alpha_j v_i - p_i) \notag\\
  &\Leftrightarrow \sum_{i \in N} (\sum_{j \in M}(x'_{i,j}-x_{i,j})\alpha_j v_i) 
                   \ge 
                   \sum_{i \in N} p'_i - \sum_{i \in N} p_i \notag\\
  &\Rightarrow     \sum_{i \in N} \delta_i v_i 
                   \ge 
                   \sum_{i \in N} p'_i - \sum_{i \in N} p_i 
                   \ge 0. \label{eq:star}
\end{align}

If $u'_i > u_i$ for some $i \in N$, then $\sum_{i \in N} u'_i > \sum_{i \in N} 
u_i$ and, thus, the first inequality in (\ref{eq:star}) is strict. Otherwise,
if $\sum_{i \in N} p'_i > \sum_{i \in N} p_i$, then the second inequality in 
(\ref{eq:star}) is strict. In both cases strictness of the inequality implies 
that $\sum_{i \in N} \delta_i v_i > 0$.

\section{Proof of \propref{prop:ic-sing-dim}}
We begin by showing that if $M$ satisfies VM and PI, then it satisfies IC. For a 
contradiction assume that $M$ satisfies VM and PI, but that it does {\em not}
satisfy IC. Then there exists $i \in N$, $\theta_i=(v_i,b_i)$, $\theta'_i=(v'_i,
b_i)$, and $\theta_{-i} = (v_{-i}, b_{-i})$ with $v_i \neq v'_i$ such that
\begin{align*}
  u_i(x_i(\theta'_i,\theta_{-i}), p(\theta'_i,\theta_{-i}), \theta_i) 
  > u_i(x_i(\theta_i,\theta_{-i}), p(\theta_i,\theta_{-i}), \theta_i).
\end{align*}

Let $c_{\gamma_{t}}(b_i,\theta_{-i}) \le v_i \le c_{\gamma_{t+1}}(b_i,
\theta_{-i})$ and let $c_{\gamma_{t'}}(b_i,\theta_{-i}) \le v'_i \le 
c_{\gamma_{t'+1}}(b_i,\theta_{-i}).$

If $v_i > v'_i$ then since $M$ satisfies VM and PI the utilities $u_i$ and 
$u'_i$ that agent $i$ gets from reports $\theta_i$ and $\theta'_i$ satisfy
\begin{align*}
  u_i - u'_i 
  &= (\gamma_{t} - \gamma_{t'}) v_i - \sum_{s=t'+1}^{t}(\gamma_s - \gamma_{s-1}) 
      c_{\gamma_s}(b_i,\theta_{-i})\\ 
  &\ge (\gamma_t - \gamma_{t'}) v_i - \sum_{s=t'+1}^{t} (\gamma_s - \gamma_{s-1}) 
      v_i\\ 
  &= 0.
\end{align*}

If $v_i < v'_i$ then since $M$ satisfies VM and PI the utilities $u'_i$ and 
$u_i$ that agent $i$ gets from reports $\theta'_i$ and $\theta_i$ satisfy
\begin{align*}
  u_i' - u_i 
  &= (\gamma_{t'}-\gamma_t) v_i - \sum_{s=t+1}^{t'}(\gamma_s - \gamma_{s-1}) 
      c_{\gamma_s}(b_i,\theta_{-i})\\
  &\le (\gamma_{t'} - \gamma_t) v_i - \sum_{s=t+1}^{t'} (\gamma_s - \gamma_{s-1}) 
      v_i\\
  &= 0.
\end{align*}

We conclude that in both cases agent $i$ is weakly better off when he reports 
truthfully. This contradicts our assumption that $M$ does {\em not} satisfy IC.

Next we show that if $M$ satisfies IC, then it satisfies VM. By contradiction
assume that $M$ satisfies IC, but that it does {\em not} satisfy VM. Then there
exists $i \in N$, $\theta_i = (v_i,b_i)$, $\theta'_i = (v'_i,b_i)$, and 
$\theta_{-i} = (v_{-i},b_{-i})$ with $v_i < v'_i$ such that
\begin{align*}
  \sum_{j \in M} x_{i,j}(\theta_i,\theta_{-i}) \alpha_j 
  > \sum_{j \in M} x_{i,j}(\theta'_i,\theta_{-i}) \alpha_j.
\end{align*}
Since $M$ satisfies IC agent $i$ with type $\theta_i$ does {\em not} benefit 
from reporting $\theta'_i$, and vice versa. Thus,
\begin{align*}
  &\sum_{j \in M} x_{i,j}(\theta_i,\theta_{-i}) \alpha_j v_i - p_i(\theta_i,
   \theta_{-i}) 
   \ge 
   \sum_{j \in M} x_{i,j}(\theta'_i,\theta_{-i}) \alpha_j v_i - p_i(\theta'_i,
   \theta_{-i}), &&\text{and}\\
  &\sum_{j \in M} x_{i,j}(\theta'_i,\theta_{-i}) \alpha_j v'_i - p_i(\theta'_i,
   \theta_{-i}) 
   \ge 
   \sum_{j \in M} x_{i,j}(\theta_i,\theta_{-i}) \alpha_j v'_i - p_i(\theta_i,
   \theta_{-i}). 
\end{align*}
By combining these inequalities we get 
\begin{align*}
  &(\sum_{j \in M} x_{i,j}(\theta_i,\theta_{-i}) \alpha_j - \sum_{j \in M} 
  x_{i,j}(\theta'_i,\theta_{-i}) \alpha_j) (v_i-v'_i) \ge 0.
\end{align*}
Since $\sum_{j \in M} x_{i,j}(\theta_i,\theta_{-i}) \alpha_j > \sum_{j \in M} 
x_{i,j}(\theta'_i,\theta_{-i}) \alpha_j$ this shows that $v_i \ge v'_i$ and 
gives a contradiction to our assumption that $v_i < v'_i.$

We conclude the proof by showing that if $M$ satisfies IC, then it satisfies PI.
For a contradiction assume that $M$ satisfies IC, but that it does {\em not} 
satisfy PI. Then there exists $i\in N$, $\theta'_i=(v'_i,b_i)$, and $\theta_{-i}
=(v_{-i},b_{-i})$ with $c_{\gamma_{t'}} \le v'_i \le c_{\gamma_{t'+1}}$ such 
that 
\begin{align*}
  p_i(\theta'_i,\theta_{-i}) \neq p_i((0,b_i),\theta_{-i}) + \sum_{s=1}^{t'}
  (\gamma_s - \gamma_{s-1}) c_{\gamma_s}(b_i,\theta_{-i}),
\end{align*}
where the $\gamma_s$ are the sum over the $\alpha$'s of all possible assignments 
in non-increasing order and the $c_{\gamma_s}(b_i,\theta_{-i})$ are the smallest 
valuations (or critical valuations) that make agent $i$ win $\gamma_s$.

Consider the smallest $v'_i$ such that this is the case. For this $v'_i$ we must 
have $v'_i=c_{\gamma_{t'}}(b_i,\theta_{-i})>c_{\gamma_0}(b_i,\theta_{-i})=0$. We 
must have $v'_i = c_{\gamma_{t'}}(b_i,\theta_{-i})$ because by VM agent $i$'s 
assignment for all reports $\theta''_i=(v''_i,b_i)$ with $v''_i$ such that 
$c_{\gamma_{t'}}(b_i,\theta_{-i})\le v''_i\le c_{\gamma_{t'+1}}(b_i,\theta_{-i})$ 
is the same and, thus, by IC he must face the same payment. We must have 
$c_{\gamma_{t'}}(b_i,\theta_{-i})>c_{\gamma_0}(b_i,\theta_{-i})=0$ because for 
$v'_i=0$ we have $p(\theta'_i,\theta_{-i})=p((0,b_i),\theta_{-i})$ by definition.

{\em Case~1:} $p_i(\theta'_i,\theta_{-i}) > p_i((0,b_i),\theta_{-i}) + \sum_{s=
              1}^{t'} (\gamma_s - \gamma_{s-1}) c_{\gamma_s}(b_i,\theta_{-i})$

Consider $\theta_i=(v_i,b_i)$ with $v_i < v'_i$ such that $c_{\gamma_{t'-1}}(b_i,
\theta_{-i}) \le v_i \le c_{\gamma_{t'}}(b_i,\theta_{-i})$. Since $v_i < v'_i$ 
we have $p_i(\theta_i,\theta_{-i})=p_i((0,b_i), \theta_{-i}) + \sum_{s=1}^{t'-1} 
(\gamma_s - \gamma_{s-1}) c_{\gamma_s}(b_i,\theta_{-i})$. If agent $i$'s type is 
$\theta'_i$ then for the utilities $u'_i$ and $u_i$ that he gets for reports 
$\theta'_i$ and $\theta_i$ we have
\begin{align*}
  u'_i - u_i 
  &< (\gamma_{t'} - \gamma_{t'-1}) v'_i - (\gamma_{t'} - \gamma_{t'-1}) 
     c_{\gamma_{t'}}(b_i,\theta_{-i}) = 0.
\end{align*} 
This shows that agent $i$ with type $\theta'_i$ has an incentive to misreport 
his type as $\theta_i$ and contradicts our assumption that $M$ satisfies IC.

{\em Case~2:} $p_i(\theta'_i,\theta_{-i}) < p_i((0,b_i),\theta_{-i}) + \sum_{s 
              = 1}^{t'} (\gamma_s - \gamma_{s-1}) c_{\gamma_s}(b_i,\theta_{-i})$

Let $\epsilon = p_i((0,b_i),\theta_{-i})+\sum_{s=1}^{t'}(\gamma_s-\gamma_{s-1}) 
c_{\gamma_s}(b_i,\theta_{-i}) - p_i(\theta'_i,\theta_{-i})$ and consider 
$\theta_i=(v_i,b_i)$ with $v_i < v'_i$ such that $c_{\gamma_{t'-1}}(b_i,
\theta_{-i}) \le v_i \le c_{\gamma_{t'}}(b_i,\theta_{-i})$. Since $v_i < v'_i$ 
we have $p_i(\theta_i,\theta_{-i}) = p_i((0,b_i),\theta_{-i})+\sum_{s=1}^{t'-1} 
(\gamma_s - \gamma_{s-1}) c_{\gamma_s}(b_i,\theta_{-i})$. If agent $i$'s type is 
$\theta_i$ then for the utilities $u'_i$ and $u_i$ that he gets from reports 
$\theta'_i$ and $\theta_i$ we have
\begin{align*}
  u'_i - u_i
  &= (\gamma_{t'} - \gamma_{t'-1}) v_i - (\gamma_{t'} - \gamma_{t'-1}) 
     c_{\gamma_{t'}}(b_i,\theta_{-i}) + \epsilon
\end{align*} 
Since this is true for all $v_i$ with $c_{\gamma_{t'-1}}(b_i,\theta_{-i})\le v_i
\le c_{\gamma_{t'}}(b_i,\theta_{-i})$ we can choose $v_i$ such that $(\gamma_{t'}
-\gamma_{t'-1}) (v_i-c_{\gamma_{t'}}(b_i,\theta_{-i})) > -\epsilon$. We get 
$u'_i - u_i > 0.$ This shows that agent $i$ with type $\theta_i$ has an 
incentive to misreport his type as $\theta'_i$ and contradicts our assumption 
that $M$ satisfies IC.

\section{Proof of Proposition~\ref{prop:utilities}}
First suppose that the payments are deterministic. If $p_i > b_i$ then 
$\ptilde_i > b_i$ and $u_i(x_i,p_i,(v_i,b_i)) = u_i(\xtilde_i,\ptilde_i,
(\vtilde_i,b_i))) = -\infty$. Otherwise, 
\begin{align*}
u_i(x_i,p_i,(v_i,b_i)) = \sum_{j=1}^{m} (x_{i,j} \alpha_j v_i) - p_i
                       = \xtilde_i \vtilde_i - \ptilde_i
                       = u_i(\xtilde_i,\ptilde_i,(\vtilde_i,b_i))).
\end{align*}

Next suppose that the payments are randomized. If $\Pr[p_i > b_i] > 0$ then 
$\Pr[\ptilde_i > b_i] > 0$ and $\E[u_i(x_i,p_i,(v_i,b_i))] = \E[u_i(\xtilde_i,
\ptilde_i,(\vtilde_i,b_i)))] = -\infty$.
Otherwise,
\begin{align*}
\E[u_i(x_i,p_i,(v_i,b_i))] = \E[\sum_{j=1}^{m} (x_{i,j} \alpha_j v_i) - p_i]
                           = \E[\xtilde_i \vtilde_i - \ptilde_i]
                           = \E[u_i(\xtilde_i,\ptilde_i,(\vtilde_i,b_i))]. 
\end{align*}

\end{document}